\documentclass[runningheads]{llncs}
\usepackage{subfig}
\usepackage{tabularx}
\usepackage{array,multirow,graphicx}
\usepackage{float}
\usepackage{color}
\usepackage{cite}
\usepackage{amssymb, extarrows}
\usepackage{amsmath}
\usepackage{amsfonts}
\usepackage{amscd}
\usepackage{enumerate}
\usepackage{url}
\usepackage{cancel}
\usepackage[normalem]{ulem}
\usepackage{algorithm, algorithmic}
\usepackage{upquote}
\usepackage{multicol}

\usepackage{color}
\usepackage[export]{adjustbox}
\usepackage{parcolumns}

\usepackage{mwe}

\def\Z{\mathbb{Z}}

\usepackage{color}
\usepackage{listings} 
\usepackage{xspace}  
\usepackage[misc,geometry]{ifsym}

\lstdefinelanguage{Sage}[]{Python}
{morekeywords={True,False,sage,singular},
	sensitive=true}
\definecolor{dblackcolor}{rgb}{0.0,0.0,0.0}
\definecolor{dbluecolor}{rgb}{.01,.02,0.7}
\definecolor{dredcolor}{rgb}{0.8,0,0}
\definecolor{dgraycolor}{rgb}{0.30,0.3,0.30}

\newcommand{\dblack}{\color{dblackcolor}\bf}

\renewcommand{\emph}[1]{{\dblack{#1}}}

\newcommand{\keywords}[1]{\par\addvspace\baselineskip
	\noindent\keywordname\enspace\ignorespaces#1}

\begin{document}
	\title{\textbf{Trapdoor Delegation and HIBE from Middle-Product LWE in Standard Model}}
\titlerunning{Trapdoor Delegation  and \textsf{DMPLWE}-based \textsf{HIBE} }

	\author{Huy Quoc Le\inst{1,2}$^\textrm{(\Letter)}$, Dung Hoang Duong\inst{1}$^\textrm{(\Letter)}$, Willy Susilo\inst{1}, Josef Pieprzyk\inst{2,3}}
	
	\authorrunning{H. Q. Le, D. H. Duong, W. Susilo and J. Pieprzyk}
	\institute{
		Institute of Cybersecurity and Cryptology, School of Computing and Information Technology, University of Wollongong, 
		Northfields Avenue, Wollongong NSW 2522, Australia.\\
			\email{qhl576@uowmail.edu.au}, \email{\{hduong,wsusilo\}@uow.edu.au} 
	\and
	CSIRO Data61, Sydney, NSW, Australia.\\
	\and 
	Institute of Computer Science, Polish Academy of Sciences, Warsaw, Poland.\\
	\email{josef.pieprzyk@data61.csiro.au } 
	}

	\maketitle
\begin{abstract}
At CRYPTO 2017, Ro{\c{s}}ca, Sakzad,  Stehl\'e and Steinfeld introduced the Middle--Product LWE (MPLWE) assumption which is as secure as Polynomial-LWE for a large class of polynomials, making the corresponding cryptographic schemes more flexible in choosing the underlying polynomial ring in design while still keeping the equivalent efficiency.
Recently at TCC 2019, Lombardi, Vaikuntanathan and Vuong introduced a variant of MPLWE assumption and constructed the first IBE scheme based on MPLWE. Their core technique is to construct lattice trapdoors compatible with MPLWE in the same paradigm of Gentry, Peikert and Vaikuntanathan at STOC 2008. However, their method cannot directly offer a Hierachical IBE construction. In this paper, we make a step further by proposing a novel trapdoor delegation mechanism for an extended family of polynomials from which we construct, for the first time, a Hierachical IBE scheme from MPLWE. Our Hierachy IBE scheme is provably secure in the standard model.
\end{abstract}


\section{Introduction}

 
Hierarchical identity-based encryption (\text{HIBE}) \cite{HL02, GS02a} is a variant of \text{IBE} \cite{Sha85},
which embeds a directed tree. The nodes of the tree are identities and
the children identities are produced by appending extra information to their parent identities. 
HIBEs can be found in many applications such as
forward-secure encryption \cite{CHK03}, broadcast encryption \cite{DF03, YFDL04}
and access control to pervasive computing information \cite{HS05} to name a few most popular.
 
In lattice-based cryptography, a crucial tool for constructing \text{IBE} and \text{HIBE} schemes is a trapdoor.
The GPV construction, for instance, applies \textit{trapdoor preimage sampleable functions} \cite{GPV08}.
The trapdoor plays a role of master secret key that is used to sample 
private key for each identity (following a distribution that is negligibly close to uniform).
This trapdoor is applied by Gentry et al.~\cite{GPV08} 
to construct their IBE from lattices in the random oracle model.
Using the same paradigm as~\cite{GPV08}, 
Agrawal et al.~\cite{AB09} introduced their IBE scheme in the standard model. 
Cash et al. \cite{CHKP10} define \textit{bonsai tree} with four basic principles in delegating a lattice basis 
(i.e., delegating a trapdoor in the \cite{GPV08} sense). 
The bonsai tree technique helps to resolve some open problems in lattice-based cryptography. 
It allows us to construct some lattice-based primitives in the standard model (without random oracles) 
as well as it facilitates delegation for purposes such as lattice-based 
\text{HIBE} schemes. 
At the same time, Agrawal et al.~\cite{ABB10} proposed two distinct trapdoor delegations 
following the definition of trapdoor from \cite{GPV08}. 
Their techniques have been used to construct a \text{HIBE} scheme in the standard model, 
which is more efficient than the one from~\cite{CHKP10}. 
Micciancio and Peikert in their work~\cite{MP12} introduced a simpler and 
more efficient trapdoor generation and delegation mechanism.

The middle-product learning with errors problem (\text{MPLWE}) is a variant of the polynomial learning with error problem (\text{PLWE}) 
proposed by Ro{\c{s}}ca et al. \cite{RSSS17}.
It exploits the middle-product of polynomials modulo $q$. 
The authors of  \cite{RSSS17} have proved that \text{MPLWE} is as secure as \text{PLWE} for a large class of polynomials.
This allows more flexibility in choosing underlying polynomial rings when designing cryptosystems.
In \cite{RSSS17}, the authors have constructed a Regev-type public key encryption scheme based on \text{MPLWE}, which is as efficient as that built over Ring-LWE~\cite{LPR10}. 
Recently, Lombardi et al. \cite{LVV19} have generalized \text{MPLWE} 
and call it degree-parametrized \text{MPLWE} (\text{DMPLWE}).
They have proved that \text{DMPLWE} is as hard as  $\text{PLWE}$ 
using similar arguments as in \cite{RSSS17}.  
Further, the authors of \cite{LVV19} have introduced a lattice trapdoor construction (following the trapdoor notion of \cite{MP12}) for \text{DMPLWE}.
The construction can be used to design a dual Regev encryption. The dual encryption allows the authors of \cite{LVV19} to come up with \text{IBE} constructions  in both the random oracle model and the standard model. The standard model \text{IBE} in \cite{LVV19} is adapted from the  framework of \cite{AB09}. 
However, a \text{DMPLWE}-based construction for a standard model \text{HIBE} 
cannot be directly obtained from the standard model \text{IBE} of \cite{LVV19}.
Thus there is a need for more work in order to define
an appropriate trapdoor delegation mechanism for the polynomial setting.
\vspace{2mm}

\noindent\textbf{Our contribution.}
In this paper, we follow the line of research initiated by the work~\cite{LVV19}.
In particular, we introduce a novel technique for delegating lattice trapdoors 
from \text{DMPLWE} and construct a new \text{HIBE} scheme based on \text{DMPLWE}.  Our \text{HIBE} scheme is provably secure in the standard model.
We follow the framework from~\cite{AB09} and~\cite{LVV19}.
 
Let $\overline{\textbf{a}}=(a_1, \cdots, a_{t'})$ be a $t'$-family of polynomials. We can interpret any polynomial as a structured matrix, e.g. Toeplitz matrix \cite{Pan01}, and hence $\overline{\textbf{a}}$ can be represented as a concatenated structured matrix, say $\textbf{A}$. 
The trapdoor from \cite{LVV19} is a modification of the trapdoor used in \cite{MP12} 
and is defined for a family of polynomials.  
More specifically, in \cite{LVV19}, a trapdoor for the  family $\overline{\textbf{a}}$ is a collection $\textsf{td}_a$ of 
\textit{short} polynomials (here \textit{short} means small coefficients), 
from which we form a matrix $\textbf{R}$ such that  $\mathbf{A}\cdot \bigl[ \begin{smallmatrix}
   \mathbf{R}\\ \mathbf{I}
 \end{smallmatrix} \bigr]=\mathbf{G},$ where $\mathbf{G}$ is the concatenated structured matrix of $\overline{\mathbf{g}}=(g_1, \cdots, g_{\gamma \tau})$, namely $g_j=2^\eta x^{d\zeta}$ for $j=\zeta \tau+\eta+1$ with $\eta \in \{0,\cdots, \tau-1\}$, $\zeta \in \{0,\cdots, \gamma-1 \}$. We call $\overline{\mathbf{g}}$ the \textit{primitive family}. 
The trapdoor $\mathsf{td}_a$ is used to search for a $t'$-family of polynomials $\overline{\textbf{r}}:=(r_1, \cdots, r_{t'})$ (following some distribution that is close to uniform) 
such that $\langle \overline{\textbf{a}},\overline{\textbf{r}} \rangle:=\sum_{i=1}^{t'}a_i\cdot r_i=u$ for any given polynomial $u$ of appropriate degree.
 
  For a construction of \text{DMPLWE}--based \text{HIBE}, we need to derive a trapdoor for an extended family of polynomials, say $\overline{\textbf{f}}=(\overline{\textbf{a}}| \overline{\textbf{h}})=(a_1, \cdots, a_{t'}| h_1, \cdots, h_{t''})$, from a trapdoor for $\overline{\textbf{a}}$. 
 To this end, we first proceed with the case $t''=\gamma \tau$, i.e., the number of polynomials in $\overline{\textbf{h}}$ has to be the same as the number in $\overline{\textbf{g}}$.
 We transform $\overline{\textbf{h}}$ into a matrix $\mathbf{H}$, and then apply the idea of trapdoor delegation from \cite{MP12} to obtain the trapdoor $\textsf{td}_f$ for $\overline{\textbf{f}}$. 
 We generalize the trapdoor delegation to the case $t''=m \gamma \tau$, a multiple of $\gamma \tau$ for $m \geq 1$. 
 
 Using the proposed polynomial trapdoor delegation, we build the first \text{HIBE} based on \text{DMPLWE}, which is provably \text{IND--sID--CPA} secure in the standard model. 
 To produce a private key for an identity $\textsf{id}=(id_1, \cdots, id_{\ell})$ at depth $\ell$, 
 we form an \textit{extended} family $\overline{\mathbf{f}}_{\mathsf{id}}=(\overline{\mathbf{a}}, \overline{\mathbf{h}}^{(1,id_1)},  \cdots,\overline{\mathbf{h}}^{(\ell,id_{\ell})})$ in which each $\overline{\mathbf{h}}^{(i,\mathsf{bit})}=(h_1^{(i,\mathsf{bit})}, \cdots,h_{t'}^{(i,\mathsf{bit})} )$ is a family of random polynomials.
Then our trapdoor delegation helps to get a trapdoor for $\overline{\mathbf{f}}_{\mathsf{id}}$, which plays the role of the private key with respect to the identity \textsf{id}. 
Deriving a private key for a child identity $\textsf{id}|id_{\ell+1}=(id_1, \cdots, id_{\ell}, id_{\ell+1})$ from a parent identity $\textsf{id}=(id_1, \cdots, id_{\ell})$ is done in similar way by appending $\overline{\mathbf{h}}^{(\ell+1,id_{\ell+1})}$ to $\overline{\mathbf{f}}_{\mathsf{id}}$  
so we get $\overline{\mathbf{f}}_{\mathsf{id}|id_{\ell+1}}=(\overline{\mathbf{a}},  \overline{\mathbf{h}}^{(1,id_1)},  \cdots,\overline{\mathbf{h}}^{(\ell,id_{\ell})},\overline{\mathbf{h}}^{(\ell+1,id_{\ell+1})})$.
Then we use the trapdoor delegation to get its private key from the private key (trapdoor) of $\overline{\mathbf{f}}_{\mathsf{id}}$. 
In order for the security proof to work, we need to put a condition on $t'$ 
such  that $t'$ is a multiple of $\gamma \tau$. 
Indeed, the condition ensures that the simulator is able to simulate an answer 
to a private key query of the adversary.
The answer is generated
 using a trapdoor for some $\overline{\mathbf{h}}^{(i,id_{\i})}$, both of which are not chosen randomly but produced by a trapdoor generator.\\
 
 \noindent\textbf{Open problems.} Our trapdoor delegation technique is restricted to the relation of the number of polynomials in the primitive family $\overline{\textbf{g}}$ and the number of polynomials in the extended family (i.e., $t''=m \gamma \tau$, a multiple of $\gamma \tau$).  It would be interesting and might be useful to find a new trapdoor delegation method that could be applied for an arbitrary $t''\geq 1$. Moreover, 
 if we had another mechanism that would help to find a trapdoor for $\overline{\mathbf{f}}'_{\mathsf{id}}$, where $\overline{\mathbf{f}}'_{\mathsf{id}}=(\overline{\mathbf{a}}, \langle \overline{\mathbf{h}}^{(1,id_1)}, \overline{\mathbf{b}} \rangle ,  \cdots, \langle \overline{\mathbf{h}}^{(\ell,id_{\ell})}, \overline{\mathbf{b}} \rangle ),$
 given a random $\overline{\mathbf{b}}$ and its trapdoor $\mathsf{td}_b$, then we might be able to apply the \text{HIBE} framework of \cite{ABB10} to get a smaller ciphertext size than that of our work here.    One more question is that whether or not there
exists a trapdoor (and delegation) method that does not utilise the Toeplitz representation
but applies directly polynomials, with a relevant definition of polynomial trapdoor.

 \noindent\textbf{Organisation.}  In Section \ref{sec2}, we review some related background. The trapdoor delegation mechanism for polynomials in MPLWE setting will be presented in Section \ref{trap}.  We will give an MPLWE-based HIBE construction in the standard model in Section \ref{sec4}. Section \ref{sec5} concludes this work.
   
\section{Preliminaries} \label{sec2}

\subsubsection{Notations.} 

 We denote by $R^{<n}[x]$ the set of polynomials of degree less than $n$ with coefficients in a commutative  ring $R$.
 We mainly work with the rings of polynomials over $\mathbb{Z}$ such as $\mathbb{Z}[x]$
 and $\mathbb{Z}_q[x]$.  We use italic small letters for  polynomials in $R$.
For a positive integer $\ell$, $[\ell]$ stands for the set $\{1, 2, \cdots , \ell\}$. 
The Gram-Schmidt orthogonal matrix of a matrix $\mathbf{A}$ is written as $\tilde{\mathbf{A}}$. 
We call $\overline{\mathbf{h}}$ an $n$-family (or $n$-vector) of polynomials 
if  $\overline{\mathbf{h}}=(h_1,\cdots, h_n)$, where $h_i$'s are polynomials. 
By $\overline{\mathbf{a}}|\overline{\mathbf{h}}$, we denote a concatenated (or expanded) family,
which consists of all ordered  polynomials from both $\overline{\mathbf{a}}$ and $\overline{\mathbf{h}}$. For two $n$-families  of polynomials $ \overline{\textbf{a}}=(a_1, \cdots, a_n)$ and $\overline{\textbf{r}}=(r_1, \cdots, c_n)$,  their scalar product is defined as $\langle \overline{\textbf{a}},\overline{\textbf{r}} \rangle:=\sum_{i=1}^{n}a_i\cdot r_i$. 
 The notation $\mathcal{U}(X)$ stands for the uniform distribution over the set $X$.
The Euclidean and sup norms of a vector $\textbf{u}$ (as well as a matrix) are
written as $\Vert \mathbf{u} \Vert$ and $ \Vert \mathbf{u} \Vert_{\infty}$, respectively.

\subsection{\text{IBE} and \text{HIBE}: Syntax and Security}
\subsubsection{Syntax.} 
An \text{IBE} system \cite{Sha85} is a tuple of algorithms \{\textsf{Setup}, \textsf{Extract}, \textsf{Encrypt}, \textsf{Decrypt}\}, in which: (1)
\textsf{Setup}($1^n$) on input a security parameter $1^n$,
		outputs a master public key $\textsf{MPK}$ and a master secret key $\textsf{MSK}$; (2) 
\textsf{Extract}(\textsf{MSK}, \textsf{id})  on input the master secret key $\textsf{MSK}$
		 and an identity $\textsf{id}$, 
		 outputs a private key $\textsf{SK}_{\textsf{id}}$; (3) 
	\textsf{Encrypt}($\textsf{MPK},\textsf{id}, \mu$) on input the master public key 
		$\textsf{MPK}$, an identity $\textsf{id}$ and a message $\mu$, outputs a ciphertext $\textsf{CT}$; and (4)
	\textsf{Decrypt}($\textsf{id}, \textsf{SK}_\textsf{id}, \textsf{CT}$)
		on input an identity $\textsf{id}$ and its associated private key 
		$\textsf{SK}_\textsf{id}$ and a ciphertet \textsf{CT}, outputs a message $\mu$.

A \text{HIBE} \cite{GS02a} is a tuple of algorithms \{\textsf{Setup}, \textsf{Extract}, \textsf{Derive}, \textsf{Encrypt}, \textsf{Decrypt}\}, 
where \textsf{Setup}, \textsf{Extract}, \textsf{Encrypt}, \textsf{Decrypt}
are defined in similar way as for \text{IBE}. 
Let $\lambda$ be the maximum depth of identities. 
An identity at depth $\ell \leq \lambda$ is represented by a binary vector $\textsf{id}=(id_{1}, \cdots, id_{\ell}) \in \{0, 1\}^\ell$ of dimension $\ell$ and it is considered as the ``parent" of the appended $\textsf{id}|id_{\ell+1}=(id_1, \cdots, id_{\ell}, id_{\ell+1})$. 
The algorithm \textsf{Setup}($1^n$, $1^\lambda$) needs a slight modification as
it accepts both $n$ and $\lambda$ as the input. 
For the input: private key $\textsf{SK}_\textsf{id}$ and  $\textsf{id}|\textsf{id}_{\ell+1}$,
the algorithm $\textsf{Derive}(\textsf{SK}_\textsf{id}, \textsf{id}|\textsf{id}_{\ell+1})$
outputs the private key $\textsf{SK}_{\textsf{id}|id_{\ell+1}}$ for the identity $\textsf{id}|id_{\ell+1}$.
If we consider the master secret key as the private key for any identity at depth $0$,  
then $\textsf{Derive}$ has the same function as  $\textsf{Extract}$.
\text{(H)IBE} has to be correct in the following sense:
$$\Pr[\textsf{Decrypt}(\textsf{id}, \textsf{SK}_{\textsf{id}}, \textsf{Encrypt}(\textsf{MPK},\textsf{id}, \mu))]=1-\textsf{negl}(n),$$
where the probability is taken over random coin tosses for \textsf{Setup}, \textsf{Extract}, \textsf{Encrypt}, \textsf{Decrypt} (for \text{IBE}) and \textsf{Derive} (for \text{HIBE}).  

\subsubsection{Security.} 
For the purpose of our paper, we present  
the following security  game for \text{IND-sID-CPA} or indistinguishability of ciphertexts under a selective chosen-identity and adaptive chosen-plaintext attack.  
In the game, the adversary  has to announce his target identity at the very beginning.
For a security parameter $n$, let $\mathcal{M}_{n}$ and $\mathcal{C}_n$  be 
the plaintext and ciphertext spaces, respectively. The game consists of six phases as follows:
\begin{itemize}
	\item \textbf{Initialize:} 
 The challenger chooses a maximum depth $\lambda$ and gives it to the adversary. The adversary outputs a target identity $\mathsf{id}^*=(id^*_1,\cdots, id^*_k), (k\leq \lambda)$.  
	\item \textbf{Setup:}
		The challenger runs $ \mathsf{Setup}( 1^{n}, 1^{\lambda})$ and sends the public parameters $\mathsf{MPK}$ to the adversary. The master secret key $\mathsf{MSK}$ is kept secret by the challenger. 
	\item \textbf{Queries 1:} The adversary makes private key queries adaptively. 
	The queries are for  identities $\mathsf{id}$ of the form $\mathsf{id}=(id_1,\cdots, id_m)$ 
	for some $m\leq \lambda$, which are not a prefix of $\mathsf{id}^*$.
	This is to say that $\mathsf{id}_i\neq \mathsf{id}^*_i$ for all $i\in[m]$ and $m\leq k$.
	The challenger answers the private key query for $\mathsf{id}$ by calling the 
	private key extraction algorithm $\mathsf{Extract}$ and sends 
	the key to the adversary.
	\item \textbf{Challenge:} 
	\begin{itemize}
		\item Whenever the adversary decides to finish Queries 1, 
			he  will output the challenge plaintext $\mu^* \in \mathcal{M}_n$.  
		\item The challenger chooses a random bit $b \in \{0,1 \}$. 
			It computes the challenge  ciphertext $\mathsf{CT}^*$. 
			If $b=0$, it calls the encryption algorithm and gets   
			$\mathsf{CT}^* \leftarrow \mathsf{Encrypt}(\mathsf{MPK},\mathsf{id}^*,\mu^*)$. 
			If $b=1$,  it chooses a random $CT \in \mathcal{C}_n$ so
			$ \mathsf{CT}^*  \leftarrow CT$. 
			$ \mathsf{CT}^*$ is then sent to the adversary.
	\end{itemize}
	\item \textbf{Queries 2:}
	The adversary makes the private key queries again and the challenger answers the queries as in \textbf{Queries 1}.
	\item \textbf{Guess:}  The adversary outputs a guess $b' \in \{0,1 \}$ and he wins if $b'=b$.

\end{itemize}
The adversary in the above game is referred to as an $\text{INDr-sID-CPA}$ adversary. The advantage of an adversary $\mathcal{A}$ in the  game is 
$\mathsf{Adv}^{\mathsf{HIBE}, \lambda,\mathcal{A}}(n)=|\Pr[b=b']-1/2|.$

\begin{definition}[$\text{IND-sID-CPA}$] A depth $\lambda$  HIBE system $\mathcal{E}$ is selective-identity indistinguishable from random if for any   probabilistic polynomial time (PPT) $\text{INDr-sID-CPA}$ adversary $\mathcal{A}$, the function $\mathsf{Adv}^{\mathsf{HIBE}, \lambda,\mathcal{A}}(n)$ is negligible. We say that  $\mathcal{E}$ is secure for 
the depth $\lambda$.
\end{definition}


\subsection{Lattices and Gaussian Distributions} \label{lattice}
 For positive integers $n, m, q$ and a matrix $\mathbf{A} \in \mathbb{Z}^{n \times m}$, 
We consider lattices 
  $  \Lambda_q^{\bot}(\mathbf{A})=\{\mathbf{z} \in \mathbb{Z}^m: \mathbf{A}\mathbf{z}=\mathbf{0} \text{ (mod } q) \}$
  $   \Lambda_q^{\mathbf{u}}(\mathbf{A})=\{\mathbf{z} \in \mathbb{Z}^m: \mathbf{A}\mathbf{z}=\mathbf{u} \text{ (mod } q) .$
If $\Lambda_q^{\mathbf{u}}(\mathbf{A}) \neq \emptyset$ then $\Lambda_q^{\mathbf{u}}(\mathbf{A})$ is a shift of $\Lambda_q^{\bot}(\mathbf{A})$. Specifically, if there exists $\mathbf{e}$ such that $\mathbf{A}\mathbf{e}=\mathbf{u} \mbox{ (mod } q)$ then $\Lambda_q^{\mathbf{u}}(\mathbf{A})=\Lambda_q^{\bot}(\mathbf{A})+\mathbf{e}.$

\begin{definition}[Gaussian Distribution]  Given countable set $S \subset \mathbb{R}^n$ and $\sigma>0$, the discrete Gaussian distribution $D_{S,\sigma,\mathbf{c}}$ over $S$ centered at some $\mathbf{c} \in S$ with standard deviation $\sigma$ is defined as $\mathcal{D}_{S, \sigma,\mathbf{c}}(\mathbf{x}):=\rho_{\sigma,\mathbf{c} }(\mathbf{x})/\rho_{\sigma,\mathbf{c} }(S),$ where $\rho_{ \sigma,\mathbf{c}}(\mathbf{x}):=\exp(\frac{-\pi \Vert \mathbf{x}-\mathbf{v}\Vert^2 }{\sigma^2})$ and $\rho_{\sigma,\mathbf{c}}(S):=\sum_{\mathbf{x}\in S}\rho_{\sigma,\mathbf{c}}(\mathbf{x})$. If $\mathbf{c}=\mathbf{0}$, we simply write $\rho_{\sigma}$ and $D_{S, \sigma}$ instead of $\rho_{\sigma,\mathbf{0}}$, $D_{S, \mathbf{0}, \sigma}$, respectively.
	
\end{definition}

We use of the following tail bound of $D_{\Lambda,\sigma}$ for parameter $\sigma$ 
sufficiently larger than the {\it smothing parameter} $\eta_\epsilon(\Lambda)$, defined to be the the smallest real number $s$ such that $\rho_{1/s}(\Lambda^{*}\setminus\{0\})\leq \epsilon$; cf.~\cite{MR07}.

\begin{lemma}[{\cite[ Lemma 2.9]{GPV08}}] \label{lem4}
	For any $\epsilon>0$, any $\sigma \geq \eta_{\epsilon}(\mathbb{Z})$, and any $K>0$, we have
	$\Pr_{x\leftarrow D_{\mathbb{Z},\sigma, c}}[\vert x-c\vert\geq K\cdot \sigma] \leq 2 e^{-\pi K^2}\cdot \frac{1+\epsilon}{1-\epsilon}.$
	In particular, if $\epsilon \in (0, \frac{1}{2})$ and $K \geq \omega(\sqrt{\log n})$,
	 then the probability that $|x-c|\geq K\cdot\sigma$ is negligible in $n$.
\end{lemma}


\subsection{Degree-Parametrized Middle-Product Learning with Errors}

\begin{definition}[{Middle-Product, \cite[Definition 3.1]{RSSS17}}] Let $d_a$, $d_b,k,d$ be integers such that $d_a+d_b-1=2k+d$. We define the middle-product of two polynomials $a \in \Z^{<d_a}[x]$ and $b \in \Z^{<d_b}[x]$ as follows: 
\end{definition}
\begin{equation}\label{eq11}
\odot_d:  \Z^{<d_a}[x] \times \Z^{<d_b}[x] \rightarrow \Z^{<d}[x], 
(a,b) \mapsto \left\lfloor \frac{ab\mod x^{k+d}}{x^k} \right \rfloor.
\end{equation}

\begin{lemma}[{\cite[Lemma 3.3]{RSSS17}}] \label{lem3}
Let $d, k, n>0$. For all $r \in R^{<k+1}[x]$, $a\in R^{<n}[x]$, $s \in R^{<n+d+k-1}[x]$, it holds that
$r\odot_d(a \odot_{d+k}s)=(r\cdot a)\odot_{d}s. $
\end{lemma}

\begin{definition}[{\text{DMPLWE}, \cite[Definition 9]{LVV19}}] 
Let $n'>0$, $q\geq 2$, $\mathbf{d}=(d_1, \cdots, d_{t'})\in [\frac{n'}{2}]^{t'}$, and let $\chi$ be a distribution over $\mathbb{R}_q$. For $s\in \mathbb{Z}_q^{<n'-1}[x]$, we define the distribution $\mathsf{DMP}_{q,n',\mathbf{d}, \chi}(s)$ over $\prod_{i=1}^{t'}(\mathbb{Z}_q^{n'-d_i}[x] \times \mathbb{R}_q^{d_i}[x] )$ as follows:
\begin{itemize}
\item For each $i\in [t']$, sample $f_i  \xleftarrow[]{\$}\mathbb{Z}_q^{<n'-d_i}[x]$ and sample $e_i \leftarrow \chi^{d_i}[x]$ (represented as a polynomial of degree less than $d_i$).
\item Output $(f_i,\mathsf{ct}_i:=f_i\odot_{d_i}s+e_i)_{i \in [t']}$.
\end{itemize}
The degree-parametrized $\mathsf{MPLWE}$ (named $\mathsf{DMPLWE}_{q,n,\textbf{d}, \chi}$) requires to distinguish between arbitrarily many samples from $\mathsf{DMP}_{q,n',\mathbf{d}, \chi}(s)$ and the same number of samples from $\prod_{i=1}^{t'}\mathcal{U}(\mathbb{Z}_q^{n'-d_i}[x] \times \mathbb{R}_q^{d_i}[x] ).$
\end{definition}

For $\mathcal{S}> 0$, let $\mathcal{F}(\mathcal{S}, \mathbf{d}, n)$ be the set of monic polynomials $f$ in $\mathbb{Z}[x]$ with the constant coeficient coprime with $q$, that have degree $m \in \cap_{i=1}^{t'}[d_i, n-d_i]$ and satisfy $\textsf{EF}(f)<\mathcal{S}$. 
For a polynomial $f \in \mathbb{Z}[x]$ of degree $m$, $\mathsf{EF}(f)$ is the {\it expansion factor} (\cite{LM06}) of $f$ defined as follows:
$\mathsf{EF}(f):=\max_{g\in \mathbb{Z}^{<2m-1}[x]}\frac{\Vert g\mod f\Vert_{\infty}}{\Vert g \Vert_{\infty}}.$
Following \cite{RSSS17}, Lombardi et al. \cite{LVV19} showed that $\mathsf{DMPLWE}$ is as hard as $\mathsf{PLWE}_{q, \chi}^{(f)}$ (defined below) for any polynomial $f$ of $\textsf{poly}(n)$-bounded expansion factor.

\begin{definition}[{\text{PLWE},\cite{SSTX09}}] 
Let $n>0$, $q\geq 2$, $f$ be a polynomial of degree $m$, $\chi$ be a distribution over $\mathbb{R}[x]/f$. The decision problem $\mathsf{PLWE}^{(f)}_{q, \chi}(s)$ is to distinguish between arbitrarily many samples 
$\{(a, a\cdot s+e): a \xleftarrow[]{\$}\mathbb{Z}_q[x]/f, e\leftarrow \chi\},$
and the same number of samples from $\mathcal{U}(\mathbb{Z}_q[x]/f \times \mathbb{R}_q[x]/f)$ over the randomness of $s \xleftarrow[]{\$}\mathbb{Z}_q[x]/f$.
\end{definition}
It is proven that $\mathsf{PLWE}^{(f)}_{q, \chi}(s)$ is as hard as solving Shortest Vector Problem (SVP) over ideal lattices in $\mathbb{Z}[x]/f$; see~\cite{SSTX09} for more detail.

\begin{theorem}[{Hardness of \text{DMPLWE}, \cite[Theorem 2]{LVV19}}] Let $n'>0$, $q\geq 2$, $\mathbf{d}=(d_1, \cdots, d_{t'})\in [\frac{n'}{2}]^{t'}$, and $\alpha \in (0,1)$. Then, there exists a probabilistic polynomial time (PPT) reduction from $\mathsf{PLWE}^{(f)}_{q, D_{\alpha \cdot q}}$ for any polynomial $f$ in $\mathcal{F}(\mathcal{S}, \mathbf{d}, n)$ to $\mathsf{DMPLWE}_{q,n',\textbf{d}, D_{\alpha' \cdot q}}(s)$ with $\alpha'=\alpha\mathcal{S} \sqrt{\frac{n'}{2}}$. 

\end{theorem}

\subsection{Lattice Trapdoor Generation for DMPLWE} \label{trapdoor}

\begin{definition}[{G-Trapdoor, \cite[Definition 5.2]{MP12}}] \label{gtrap}
Let $\mathbf{A} \in \mathbb{Z}_q^{n \times m}$ and $\mathbf{G} \in \mathbb{Z}_q^{n \times m'}$ be matrices with $m \geq m' \geq n$. A matrix $\mathbf{R}\in \mathbb{Z}^{(m-m')\times m'}$ is called $\mathbf{G}$-trapdoor for $\mathbf{A}$ with tag $\mathbf{H}$ (which is an invertible matrix in $\mathbb{Z}_q^{n \times n}$) if
 $\mathbf{A}\cdot \begin{bmatrix}
   \mathbf{R}\\ \mathbf{I}_{m'}
 \end{bmatrix}=\mathbf{H}\mathbf{G}.$
\end{definition}

In particular,  it is suggested in \cite[Section 4]{MP12} that $\mathbf{G}=\mathbf{I}_n\otimes \begin{bmatrix}
1& 2&\cdots& 2^{k} 
\end{bmatrix}$. We can choose $\mathbf{H}=\mathbf{I}_n$ or such that $\mathbf{H}\mathbf{G}$ is any (column) permutation of $\mathbf{G}$ which is similar to the usage of $\mathbf{G}$ in \cite{LVV19}. In fact, it is defined in \cite[Definition 11]{LVV19}) that $\mathbf{A} \in \mathbb{Z}^{k \times (m+k\tau)}$ and  $\mathbf{G}:=\mathbf{I}_k \otimes[ 1 \mbox{ } 2 \cdots 2^{\tau-1}] \in \mathbb{Z}_q^{k \times  k\tau}$. However, $\mathbf{G}$ is used in $\mathsf{SamplePre}$ (see below) is actually a (column) permutation of $\mathbf{I}_k \otimes[ 1 \mbox{ } 2 \cdots 2^{\tau-1}] $ from which the authors can extracts polynomial $g_i$ in $\overline{\textbf{g}}$ thanks to the Toeplitz representation of polynomials (see Equation \eqref{k2}). We first recall their definition and some basic properties.

\begin{definition}[Toeplitz matrix] \label{top}
	Let $R$ be a ring and $d,k >0$ be integers.
	For any polynomial $u\in R^{<n}[x]$, we define the Topelitz matrix $\mathsf{Tp}^{n,d}(u)$ for $u$ as a matrix in $R^{(n+d-1) \times d}$ whose the $i$-th column is the coefficient vector of $x^{i-1}\cdot u$ arranged in increasing degree of $x$ with $0$ inserted if any. 
	
\end{definition}

By Definition \ref{top}, it is easy to assert the following Lemma.

\begin{lemma}\label{lem2} Let $u \in \mathbb{Z}^{<n}[x]$. Then,
	$$\mathsf{Tp}^{n,d}(u)=[\mathsf{Tp}^{n+d-1,1}(u)|\mathsf{Tp}^{n+d-1,1}(x\cdot u)|\cdots |\mathsf{Tp}^{n+d-1,1}(x^{d-1}\cdot u)].$$
	
\end{lemma}

\begin{lemma}[{\cite[ Lemma 7]{LVV19}}] \label{lem1}
	For positive integers $k, n, d$ and polynomials $u\in R^{<k}[x]$, 
	if $v\in R^{<n}[x]$, then
	$\mathsf{Tp}^{k,n+d-1}(u)\cdot \mathsf{Tp}^{n,d}(v)= \mathsf{Tp}^{k+n-1,d}(u\cdot v).$
\end{lemma} 

\begin{theorem}[{\cite[Theorem 4]{LVV19}}]\label{kk}
	Let $\mathbf{G}:=\mathbf{I}_k \otimes \begin{bmatrix}
	1& 2&\cdots& 2^{\tau-1} 
	\end{bmatrix} \in \mathbb{Z}_q^{k \times  k\tau}$ and matrices $\mathbf{A} \in \mathbb{Z}^{k \times (m+k\tau)}$,  $\mathbf{R}\in \mathbb{Z}^{m\times k \tau}$ be such that
	$\mathbf{A}\cdot \bigl[ \begin{smallmatrix}
	\mathbf{R}\\ \mathbf{I}_{k \tau}
	\end{smallmatrix} \bigr]=\mathbf{G}.$
Then, there exists an efficient algorithm $\mathcal{P}=(\mathcal{P}_1, \mathcal{P}_2)$ that executes according to the two following phases: 
\begin{itemize}
	\item \textbf{offline}: $\mathcal{P}_1(\mathbf{A}, \mathbf{R}, \sigma)$ performs some polynomial-time preprocessing on input $(\mathbf{A}, \mathbf{R}, \sigma)$ and outputs a state $\mathsf{st}$.
	\item \textbf{online}: for a given vector  $\mathbf{u}$, 
		$\mathcal{P}_2(\mathsf{st}, \mathbf{u})$ samples a vector from 
		$D_{\Lambda^{\bot}_{\textbf{u}}(\textbf{A}),\sigma}$ as long as
	\begin{equation}\label{k145}
	\sigma \geq \omega(\sqrt{\log k})\cdot \sqrt{7(s_1(\mathbf{R})^2+1)}, 
	\end{equation}
	 where $s_1(\mathbf{R}):=\max_{\Vert  \mathbf{u}\Vert=1} \Vert \mathbf{R} \mathbf{u}\Vert$ is the largest singular value of $\mathbf{R}$.
\end{itemize}
\end{theorem}
The value $s_1(\mathbf{R})$ is upper bounded as explained by the lemma given below.
\begin{lemma}[{\cite[ Lemma 6]{LVV19}}] \label{lem5}
	For any matrix $\mathbf{R}=(R_{ij}) \in \mathbb{R}^{m \times n}$,  \begin{equation}\label{k16}
	s_1(\mathbf{R}) \leq \sqrt{mn}\cdot \max_{i,j}\vert R_{ij}\vert.
	\end{equation}
	\end{lemma}

\noindent\textbf{G-Trapdoor for a family of polynomials.}
 We recap the construction of lattice trapdoors for {DMPLWE} from \cite{LVV19}. The construction applies two PPT algorithms $\mathsf{TrapGen}$ and $\mathsf{SamplePre}$.Suppose that $q=\textsf{poly}(n)$, $d \leq n$, $dt/n=\Omega(\log n)$, $d\gamma=n+2d-2$, $\tau:=\lceil \log_2 q \rceil$,  $\beta:=\lceil \frac{\log_2(n)}{2} \rceil  \ll q/2$, and $\sigma$ satisfies Equation \eqref{k12} below. Then, \textsf{TrapGen} and\textsf{ SamplePre} work as follows:  
\begin{description}
\item 
$\underline{\mathsf{TrapGen}(1^n)}$: On input a security parameter $n$, do the following:
\begin{itemize}
\item Sample $\overline{\mathbf{a}}'=(a_1,\cdots, a_t) \xleftarrow[]{\$} (\mathbb{Z}_q^{<n}[x])^t$, and for all $j\in [\gamma\tau]$, sample $\overline{\mathbf{w}}^{(j)}=(w^{(j)}_1, \cdots, w^{(j)}_t)\leftarrow (\Gamma^d[x])^t$ where $\Gamma=\mathcal{U}(\{-\beta,\cdots,\beta\})$.
\item For all $j\in [\gamma\tau]$, define
$u_j=\langle \overline{\mathbf{a}}',\overline{\mathbf{w}}^{(j)} \rangle$ 
and $a_{t+j}=g_j-u_j,$
 where \begin{equation}\label{k39}
 g_j=2^\eta x^{d\zeta} \in \mathbb{Z}_q^{n+d-1}[x],
 \end{equation} for $j=\zeta \tau+\eta+1$ with $\eta \in \{0,\cdots, \tau-1\}$, $\zeta \in \{0,\cdots, \gamma-1 \}$. Set $\overline{\mathbf{g}}:=(g_1,\cdots, g_{\gamma \tau})$.
 
\item Output $\overline{\mathbf{a}}:=(a_1,\cdots, a_t, a_{t+1},\cdots, a_{t+\gamma\tau})$ with its corresponding trapdoor $\mathsf{td}:=(\overline{\mathbf{w}}^{(1)}, \cdots, \overline{\mathbf{w}}^{({\gamma \tau})})$.
\end{itemize}
The amount of space to store the trapdoor $\mathsf{td}$ is $O(d(\gamma\tau)t)=O(n\tau t)$ as $d\gamma=n+2d-2 \leq 3n.$

\item 
$\underline{\mathsf{SamplePre}(\overline{\mathbf{a}}=(a_1, \cdots, a_{t+\gamma\tau}),\mathsf{td}=(\overline{\mathbf{w}}^{(1)}, \cdots, \overline{\mathbf{w}}^{({\gamma \tau})}),u,\sigma)}$: On input a family $\overline{\mathbf{a}}$ of $t+\gamma\tau$ polynomials together with its trapdoor $\mathsf{td}_{\epsilon}$ generated by \textsf{TrapGen}, and a polynomial $u$ of degree less than $n+2d-2$, do the following: 
\begin{itemize}
\item First, construct (implicitly) matrices $\textbf{A}',\textbf{A},\textbf{T}, \textbf{G}$ for $\overline{\mathbf{a}}'$, $\overline{\mathbf{a}}$, $\mathsf{td}$, $\overline{\mathbf{g}}$, respectively:
$$\mathbf{A}'=[\mathsf{Tp}^{n,2d-1}(a_1)|\cdots|\mathsf{Tp}^{n,2d-1}(a_t)],$$
	\begin{equation*}
	\mathbf{A}=[\mathsf{Tp}^{n,2d-1}(a_1)|\cdots|\mathsf{Tp}^{n,2d-1}(a_t)| \mathsf{Tp}^{n+d-1,d}(a_{t+1})|\cdots|\mathsf{Tp}^{n+d-1,d}(a_{t+\gamma \tau})],
		\end{equation*}

	\begin{equation}\label{k34}
\mathbf{T}= \begin{bmatrix}
\mathsf{Tp}^{d,d}(w^{(1)}_1) & \cdots & \mathsf{Tp}^{d,d}(w^{(\gamma\tau)}_1 )\\
\vdots &  & \vdots\\
\mathsf{Tp}^{d,d}(w^{(1)}_t) & \cdots & \mathsf{Tp}^{d,d}(w^{(\gamma\tau)}_t )
\end{bmatrix} \in \mathbb{Z}_q^{(2d-1)t \times d\gamma \tau},
	\end{equation}

	\begin{equation}\label{k2}
\mathbf{G}=[\mathsf{Tp}^{n+d-1,d}(g_1)|\cdots|\mathsf{Tp}^{n+d-1,d}(g_{\gamma \tau})] \in \mathbb{Z}_q^{d\gamma \times d\gamma \tau},
\end{equation}

$$\mathbf{I}_{d\gamma\tau}=\begin{bmatrix}
\mathsf{Tp}^{1,d}(1) & \cdots & \\
\cdots &  & \cdots\\
 & \cdots & \mathsf{Tp}^{1,d}(1)
\end{bmatrix}  \in \mathbb{Z}_q^{d\gamma \tau \times d\gamma \tau}.$$
Then  $\mathbf{A}=[\mathbf{A}'|\mathbf{G}-\mathbf{A}'\mathbf{T}]$ and hence $\mathbf{A}\cdot \bigl[ \begin{smallmatrix}
\mathbf{T}\\ \mathbf{I}_{d\gamma\tau}
\end{smallmatrix} \bigr]=\mathbf{G}$. Recall that $d\gamma=n+2d-2.$

\item The polynomial $u$ is represented it as $\textbf{u}=\mathsf{Tp}^{n+2d-2,1}(u) \in \mathbb{Z}_q^{n+2d-2}$.
\item Sample vector $\textbf{r} \in \mathbb{Z}^{(2d-1)t+d\gamma\tau}$ from $D_{\Lambda^{\bot}_{\textbf{u}}(\textbf{A}),\sigma}$ using the trapdoor $\textbf{T}$ in means of \cite{MP12}, where 
\begin{equation}\label{k12}
\sigma \geq \omega(\sqrt{\log(d \gamma)})\cdot \sqrt{7(s_1(\mathbf{T})^2+1)},
\end{equation}
and
\begin{equation}\label{k19}
s_1(\mathbf{T}) \leq \sqrt{(2d-1)t \cdot (d\gamma \tau)}\cdot \beta.
\end{equation}

\item Split $\textbf{r}$ into $\textbf{r}=[\textbf{r}_1^{\top}|\cdots| \textbf{r}^{\top}_{t+\gamma \tau}]^{\top}$, and rewrite it (in column) as a Toeplitz matrix of polynomials $r_1, \cdots, r_{t+\gamma\tau}$, where $\textbf{r}_j=\mathsf{Tp}^{2d-1,1}(r_{j})$,  $\deg(r_j)<2d-1$, $\forall j\in [t]$,  $\textbf{r}_j=\mathsf{Tp}^{d,1}(r_{j})$, $\deg(r_{t+j}) <d$,$\forall j\in {t+1, \cdots, t+\gamma\tau}$.

  \item Output $\overline{\mathbf{r}}:=(r_1,\cdots,r_{t+\gamma\tau})$. Note that, $\langle \overline{\mathbf{a}}, \overline{\mathbf{r}} \rangle=\sum_{i=1}^{t+\gamma  \tau} a_i\cdot r_i=u$; see \cite[Section 5]{LVV19} for more details.
\end{itemize}

The runtime of \textsf{SamplePre} is $\tilde{O}(nt)$ and the output distribution of $(r_i)$ is exactly the conditional distribution 
$$(D_{\mathbb{Z}^{2d-1}, \sigma}[x])^t \times (D_{\mathbb{Z}^{d}, \sigma}[x])^{\gamma\tau}|\sum_{i=1}^{t+\gamma  \tau} a_i \cdot r_i=u. $$
\end{description}

Further on, we give our main results, which are a trapdoor delegation mechanism useful 
for extending a family of polynomials as well as a \text{HIBE} system 
built using the framework of \cite{AB09}. 
From now on, by ``trapdoor", we mean ``\textbf{G}-trapdoor", 
where $\textbf{G}$ is defined by Equation \eqref{k2}. 
Also, we denote  the output of \textsf{TrapGen} by $\overline{\mathbf{a}}_{\epsilon}$ and $\mathsf{td}_{\epsilon}$ and call them \textit{the root family} and \textit{the root trapdoor}, respectively. 
The  Toeplitz matrices $\textbf{A}_{\epsilon}$ and $\textbf{T}_{\epsilon}$
correspond to $\overline{\mathbf{a}}_{\epsilon}$ and $\mathsf{td}_{\epsilon}$,
respectively.

\section{Trapdoor Delegation for Polynomials } \label{trap}
\subsection{Description}
   In order to exploit the trapdoor technique in constructing a \text{MPLWE}-based \text{HIBE} scheme, we have to solve the problem of delegating a trapdoor (in the sense of Definition \ref{gtrap}) for $\overline{\mathbf{f}}=(a_1,\cdots, a_{t'}|h_1, \cdots, h_{t''} )$ provided the trapdoor for $\overline{\mathbf{a}}=(a_1,\cdots, a_{t'})$. 
    As mentioned in Section \ref{trapdoor}, we can represent $\overline{\mathbf{f}}$ as a concatenation of Toeplitz matrices of the form $\mathbf{F}=[\mathbf{A}|\mathbf{H}]$ in which $\mathbf{A}, \mathbf{H}$ are the Toeplitz representations for $\overline{\mathbf{a}}$ and $\overline{\mathbf{h}}:=(h_1, \cdots, h_{t''})$, respectively. 
    
Following Definition \ref{gtrap}, our task is to find a matrix $\mathbf{R}$,
 which satisfies the equation 
    $\mathbf{F}\cdot \bigl[ \begin{smallmatrix}
       \mathbf{R}\\ \mathbf{I}
     \end{smallmatrix} \bigr]=\mathbf{G}$,
where $\mathbf{G}$ as given by Equation \eqref{k2}. 
Recall that, in matrix setting in \cite[Section 5.5]{MP12}, this task can be easily done by finding $\mathbf{R}$ that satisfies the relation $ \mathbf{A}\mathbf{R}=\mathbf{G}-\mathbf{H}$, 
when we know a trapdoor for $\mathbf{A}$ and 
 $\textbf{H}$ has the same dimension as $\textbf{G}$. 
In our setting, this task is not straightforward. 
The main reason for this is that the matrices $\textbf{A}$, $\textbf{G}$, $\textbf{H}$ 
are Toeplitz ones. 
To be able to  apply the idea of trapdoor delegation of \cite{MP12} to our setting, 
we have to design $\mathbf{H}$ such that $\mathbf{U}:=\mathbf{G}-\mathbf{H}$ is still in the Toeplitz form of some polynomials. 
In other words, the form of $\mathbf{H}$ should be similar in form and in dimension to that of $\mathbf{G}$ in \eqref{k2}, namely,
      	\begin{equation}\label{k8}
      	\mathbf{H}=[\mathsf{Tp}^{n+d-1,d}(h_1)|\cdots|\mathsf{Tp}^{n+d-1,d}(h_{\gamma \tau})] \in \mathbb{Z}_q^{d\gamma \times d\gamma \tau}. 
      	\end{equation}
This requires that $t''=\gamma \tau$ and $\deg(h_i)<n+d-1$ for all $i\in [\gamma \tau]$. 
If this is the case, the last step is to try to follow \cite{LVV19} using $\mathsf{SamplePre}$ 
to have  $\mathbf{R}$ satisfy $\mathbf{A}\mathbf{R}=\mathbf{U}$ given $\mathbf{A}$ 
and a trapdoor for $\mathbf{A}$. 
Note that in our polynomial setting $\textbf{R}$ should  be a structured matrix, which can be easily converted into appropriate polynomials $r_i$. 

     By generalization, we come up with the following theorem in which $t'=t+k\gamma\tau$ and $t''=m\gamma\tau$ for $k \geq 1, m\geq 1$:
\begin{theorem}[Trapdoor Delegation] \label{multiple}
Let $n$  be a positive integer, $q=\textsf{poly}(n)$ be a prime, and $d, t,$ $ \gamma, \tau, k$, $m$ be positive integers such that $d \leq n$, $dt/n =\Omega(\log n)$, $ d\gamma=n+2d-2$, $k\geq 1$, $m \geq 1$. Let $\tau:=\lceil \log_2 q \rceil$ and  $\beta:=\lceil \frac{\log_2 n}{2} \rceil$.	Let $\mathbf{G}$ be matrix as in \eqref{k2} and $\overline{\mathbf{a}}=(a_1,\cdots,a_{t+k\gamma \tau})$ be a ($t+k\gamma \tau$)-family of polynomials and its associated trapdoor $\mathsf{td}_a$, where $a_i \in \mathbb{Z}^{<n}_q[x] $ for $i \in [t]$ and $a_i \in \mathbb{Z}^{<n+d-1}_q[x] $ for $t+1\leq i \leq t+k\gamma \tau$. Suppose that $\overline{\mathbf{h}}=(h_1,\cdots, h_{m\gamma \tau})$ is a $m\gamma \tau$-family of polynomials in $\mathbb{Z}^{<n+d-1}_q[x]$ and $\overline{\sigma}=(\sigma_{k+1}, \cdots, \sigma_{k+m}) $ to be determined. Then, there exists an efficient (PPT) algorithm,  $\mathsf{SampleTrap}(\overline{\mathbf{a}},\overline{\mathbf{h}},\mathsf{td}_a,\overline{\sigma})$ that outputs a trapdoor $\mathsf{td}_f$ for $\overline{\mathbf{f}}=(a_1,\cdots,a_{t+k\gamma \tau}| h_1,\cdots, h_{m\gamma \tau})$. 
         Moreover, the amount of space to store the trapdoor $\mathsf{td}_f$ is $O(((2d-1)t+(k+m-1)\gamma \tau)\cdot d\gamma \tau)=O(n^2 \log^2 n)=\widetilde{O}(n^2)$.
         \end{theorem}

\subsection{Elementary Trapdoor Delegation }\label{basictrap}
In this section, we present in detail the basic trapdoor delegation
for the family $\overline{\mathbf{f}}=(a_1, \cdots, a_{t+\gamma \tau} | h_1, \cdots,h_{\gamma \tau})$ 
given the root trapdoor $\textsf{td}_{\epsilon}$ for the root family $\overline{\mathbf{a}}_{\epsilon}=(a_1, \cdots, a_{t+\gamma \tau})$.
They are generated by $\mathsf{TrapGen}$, i.e.,  $\mathsf{SampleTrap}$ for $k=1$ and $m=1$.
This process is called  \textsf{TrapDel} and is shown as Algorithm \ref{algo1}.
  
Note that $\textsf{TrapGen}$,    	
$\overline{\mathbf{a}}_{\epsilon}=(a_1,\cdots,a_t,a_{t+1},\cdots,a_{t+\gamma \tau}) \in (\mathbb{Z}_q^{<n}[x])^t \times (\mathbb{Z}_q^{<n+d-1}[x])^{\gamma \tau},$ 
and the corresponding concatenated Toeplitz matrix $\mathbf{A}_{\epsilon} \in \mathbb{Z}_q^{(n+2d-2)\times [(2d-1)t+d\gamma \tau]}$ is constructed as 
  	\begin{equation}\label{k10}
  \mathbf{A}_{\epsilon}=[\mathsf{Tp}^{n,2d-1}(a_1)|\cdots|\mathsf{Tp}^{n,2d-1}(a_t)| \mathsf{Tp}^{n+d-1,d}(a_{t+1})|\cdots|\mathsf{Tp}^{n+d-1,d}(a_{t+\gamma \tau})].
  \end{equation} 
The matrix $\textbf{G}$ has the following form:
 $$\mathbf{G}=[ \mathsf{Tp}^{n+d-1,d}(g_1) | \cdots | \mathsf{Tp}^{n+d-1,d}(g_{\gamma \tau})],$$ 
  where $g_j=2^\eta x^{d\zeta}$ for $j=\zeta \tau+\eta+1$ with $\eta \in \{0,\cdots, \tau-1\}$, $\zeta \in \{0,\cdots, \gamma-1 \}$.
As discussed above, we construct   
    $\mathbf{H}=[ \mathsf{Tp}^{n+d-1,d}(h_1) | \cdots | \mathsf{Tp}^{n+d-1,d}(h_{\gamma \tau})]$ for $h_1, \cdots, h_{\gamma \tau}$, whose $\deg(h_i)<n+d-1$ for all $i\in [\gamma \tau]$. 
Then the Toeplitz matrix for $\overline{\mathbf{f}}$ takes the form 
    	\begin{equation}\label{k30}
     \begin{split}
      \mathbf{F}&=[\mathbf{A}_{\epsilon}|\mathbf{H}]\\&=[\mathsf{Tp}^{n,2d-1}(a_1)|\cdots|\mathsf{Tp}^{n,2d-1}(a_t)| \mathsf{Tp}^{n+d-1,d}(a_{t+1})|\cdots|\mathsf{Tp}^{n+d-1,d}(h_{\gamma \tau})].
     \end{split}
      \end{equation} 
 and 
    \begin{equation}\label{k11}
    \mathbf{G}-\mathbf{H}=[ \mathsf{Tp}^{n+d-1,d}(g_1-h_1) | \cdots | \mathsf{Tp}^{n+d-1,d}(g_{\gamma \tau}-h_{\gamma \tau})].
        \end{equation}
 For $i=1,\cdots, \gamma \tau$, let $u_i=g_i-h_i$. 
From Lemma \ref{lem2}, we have
 \begin{equation*}\label{k9}
 \begin{split}
   \mathbf{G}-\mathbf{H}=&[ \mathsf{Tp}^{n+2d-2,1}(u_1) | \cdots | \mathsf{Tp}^{n+2d-2,1}(x^{d-1}\cdot (u_1)) |\\
   &\cdots |\mathsf{Tp}^{n+2d-2,1}(u_{\gamma \tau}) |\cdots |\mathsf{Tp}^{n+2d-2,1}(x^{d-1}\cdot (u_{\gamma \tau}))]\\
   &=[ \mathsf{Tp}^{n+2d-2,1}(v_1) | \cdots | \mathsf{Tp}^{n+2d-2,1}(x^{d-1}\cdot (v_{d\gamma\tau}))],
    \end{split}
  \end{equation*}
  where $v_i=x^\alpha u_\beta$ for $i=\alpha+d(\beta-1)+1$, with 
   $\alpha \in \{0,\cdots, d-1\}$, $\beta \in \{1,\cdots, \gamma \tau \}$. 
Let $\mathbf{v}^{(i)}:=\mathsf{Tp}^{n+2d-2,1}(v_i)$. Now, for $i=1,\cdots, \gamma \tau$ we have to find $\mathbf{R}=[\textbf{r}^{(1)}|\cdots|\textbf{r}^{(d\gamma \tau)}]$ such that   $\mathbf{A}_{\epsilon}[\textbf{r}^{(1)}|\cdots|\textbf{r}^{(d\gamma \tau)}]=[\textbf{v}^{(1)}|\cdots|\textbf{v}^{(d\gamma \tau)}],$ which is equivalent to $\mathbf{A}_{\epsilon}\textbf{r}^{(i)}=\textbf{v}^{(i)}$ for $1\leq i \leq d\gamma \tau$. This can be done using $\mathsf{SamplePre}(\overline{\textbf{a}}_{\epsilon},$ $\mathsf{td}_{\epsilon}, v_i,\sigma)$. 
Eventually, we get $\mathbf{r}^{(i)} \in \mathbb{Z}^{(2d-1)t+d\gamma \tau}$, which is sampled from $\mathcal{D}_{\Lambda^{\bot}_{\textbf{v}^{(i)}}(\textbf{A}),\sigma}$, where $\sigma \geq \omega(\sqrt{\log(d \gamma)})\cdot \sqrt{7((2d-1)t \cdot (d\gamma \tau)\cdot \beta^2+1)}$; see \eqref{k12}, \eqref{k19}.  
 
 Finally, we obtain  the trapdoor $\textsf{td}_f=(\overline{\mathbf{r}}^{(1)}, \cdots, \overline{\mathbf{r}}^{(d \gamma \tau)})$ for $\overline{\textbf{f}}$, where $\overline{\mathbf{r}}^{(i)}=(r^{(i)}_{1}, \cdots, r^{(i)}_{t+\gamma \tau}) $,
     with $\deg(r^{(i)}_j)<2d-1$ for $j\in[t]$, $\deg(r^{(i)}_{t+j})<d$ for $j \in [\gamma \tau]$ and for all $i \in [d \gamma \tau]$. and its corresponding matrix representation is
 \begin{equation}\label{k13}
 \textbf{R}=(R_{ij})=\begin{bmatrix}
 \mathsf{Tp}^{2d-1,1}(r^{(1)}_{1})& \cdots& \mathsf{Tp}^{2d-1,1}(r^{(d\gamma \tau)}_{1})\\
 \vdots&&\vdots\\
 \mathsf{Tp}^{2d-1,1}(r^{(1)}_{t}) & \cdots& \mathsf{Tp}^{2d-1,1}(r^{(d\gamma \tau)}_{t})\\
 \mathsf{Tp}^{d,1}(r^{(1)}_{t+1})& \cdots& \mathsf{Tp}^{d,1}(r^{(d\gamma \tau)}_{t+1})\\
 \vdots&&\vdots\\
 \mathsf{Tp}^{d,1}(r^{(1)}_{t+\gamma \tau})& \cdots& \mathsf{Tp}^{d,1}(r^{(d\gamma \tau)}_{t+\gamma \tau})\\ 
 \end{bmatrix} \in \mathbb{Z}^{((2d-1)t+d\gamma \tau) \times d\gamma \tau}.
 \end{equation}
 Certainly, we have $\mathbf{F}\cdot \bigl[ \begin{smallmatrix}
         \mathbf{R}\\ \mathbf{I}
       \end{smallmatrix} \bigr]=\mathbf{G}$.  Remark that, by Lemma \ref{lem4},
  \begin{equation}\label{k17}
  \vert R_{ij} \vert \leq \omega(\log n)\cdot  \sigma \mbox { with probability } 1-\mathsf{negl}(n).  
  \end{equation} 
  Hence, from Lemma \ref{lem5}
   \begin{equation}\label{k18}
 s_1(\mathbf{R}) \leq \sqrt{((2d-1)t+d\gamma \tau)\cdot (d\gamma \tau)}\cdot \omega(\log n) \cdot \sigma,  
  \end{equation}
   where $\sigma$ satisfies Equation \eqref{k12}.

    \begin{algorithm}[h]
    	\caption{$\mathsf{\mathsf{TrapDel}(\overline{\textbf{a}},\overline{\textbf{h}},\mathsf{td},\sigma)}$}
    	\begin{algorithmic}[1]
    		
    		\REQUIRE A $(t+k\gamma \tau)$-family of polynomials $\overline{\mathbf{a}}=(a_1,\cdots,a_t,a_{t+1},\cdots,a_{t+k\gamma \tau}) \in (\mathbb{Z}_q^{<n}[x])^t \times (\mathbb{Z}_q^{<n+d-1}[x])^{k\gamma \tau},$ and its trapdoor $\mathsf{td}_a$, and a $\gamma\tau$-family of polynomials $\overline{\textbf{h}}=(h_1, \cdots, h_{\gamma \tau}) \in  (\mathbb{Z}_q^{<n+d-1}[x])^{\gamma \tau}$, and (implicitly) $\overline{\mathbf{g}}=(g_1,\cdots, g_{\gamma \tau}) \in (\mathbb{Z}_q^{<n+d-1}[x])^{\gamma \tau}$ as in \eqref{k39}.
    		\ENSURE The trapdoor $\mathsf{td}_f$ for $\overline{\textbf{f}}=(a_1,\cdots,a_{t+k\gamma \tau},h_1, \cdots, h_{\gamma \tau})$.
    		\STATE Compute $\overline{\textbf{u}}=(u_1,\cdots, u_{\gamma \tau}) \leftarrow \overline{\textbf{g}}-\overline{\textbf{h}}=(g_1-h_1, \cdots, g_{\gamma \tau}-h_{\gamma \tau})$. 
    		\STATE Define $v_i=x^\alpha u_\beta$ for $i=\alpha+d(\beta-1)+1$, with 
    		   $\alpha \in \{0,\cdots, d-1\}$, $\beta \in \{1,\cdots, \gamma \tau \}$.
    		\STATE   For $i\in [d\gamma\tau]$, call $ \mathsf{GenSamplePre}(\overline{\textbf{a}},\mathsf{td}_a, v_i,\sigma)$ to get $\overline{\mathbf{r}}^{(i)}=(r^{(i)}_{1}, \cdots, r^{(i)}_{t+k\gamma \tau})$,
    		    where $\deg(r^{(i)}_j)<2d-1$ for $j\in[t]$, $\deg(r^{(i)}_{t+j})<d$ for $j \in [k\gamma \tau]$ and for all $i \in [d \gamma \tau]$.
    		    \STATE Return $\mathsf{td}_f=(\overline{\mathbf{r}}^{(1)}, \cdots, \overline{\mathbf{r}}^{(d\gamma\tau)})$.
    	\end{algorithmic}
    	\label{algo1}
    \end{algorithm}


 Note that, after having the trapdoor for $\overline{\mathbf{f}}$ and by assigning 
 $\overline{\mathbf{a}}_{\epsilon}\leftarrow \overline{\mathbf{f}}$, 
 $\textbf{A}_{\epsilon} \leftarrow \textbf{F}$, 
 we can perform the same procedure explained above.
So we get a trapdoor for $\overline{\mathbf{f}}'=(a_1,\cdots,a_{t+\gamma \tau}, h_1,\cdots, h_{\gamma \tau}|z_1,\cdots, z_{\gamma \tau} )$ for some $\overline{\mathbf{z}}=(z_1,\cdots, z_{\gamma \tau})$, where $z_i\in \mathbb{Z}^{<n+d-1}_q[x]$.  
Consequently, 
we come up with a PPT algorithm called $\mathsf{TrapDel}$
(Algorithm \ref{algo1}) in which we consider the expanded families of the 
form $\overline{\mathbf{a}}=(a_1,\cdots,a_t,a_{t+1},\cdots,a_{t+k\gamma \tau}) \in (\mathbb{Z}_q^{<n}[x])^t \times (\mathbb{Z}_q^{<n+d-1}[x])^{k\gamma \tau}$ for $k \geq 1$. 
Also note that \textsf{TrapDel} does not call \textsf{SamplePre}.
Instead, it calls a slightly modified variant presented below.

\subsubsection{   Generalized \textsf{SamplePre}. }     
Accordingly to the expansion of trapdoors, 
we slightly modify \textsf{SamplePre} in Section \ref{trapdoor}
and call it $\mathsf{GenSamplePre}$.
The algorithm works not only with \textsf{TrapGen} 
(i.e., $k=1$) but also with \textsf{TrapDel} (i.e., $k>1$). 
$\mathsf{GenSamplePre}$ is the same as $\mathsf{SamplePre}$ 
except for $k>1$, where matrices  
\textbf{R} given as the input trapdoors are of form \eqref{k13}, 
while for $k=1$, the matrix \textbf{R} is of form \eqref{k34}. 
If we execute    
$\mathsf{GenSamplePre}$ for  input $(\overline{\mathbf{a}}=(a_1, \cdots, a_{t+k\gamma\tau}),\mathsf{td}_a=(\overline{\mathbf{r}}^{(1)}, \cdots, \overline{\mathbf{r}}^{({d\gamma \tau})}),u,\sigma)$, where $\overline{\mathbf{r}}^{(i)}=(r^{(i)}_1, \cdots,$ $ r^{(i)}_{k\gamma\tau})$ (with $k>1$),
 then $\mathsf{td}_a$ should be interpreted as a ${((2d-1)t+(k-1)d\gamma \tau) \times d\gamma \tau}$-matrix, say $ \textbf{R}^{(k-1)}$, of the form \eqref{k13}.
The last row is indexed by $t+(k-1)\gamma \tau$.

\subsection{\textsf{SampleTrap}} \label{multrap}

{$\textsf{SampleTrap}$ mentioned in Theorem \ref{multiple} is described 
as follows:

\underline{{\textsf{SampleTrap}($\overline{\mathbf{a}}= (a_1, \cdots, a_{t+k\gamma\tau}), \overline{\mathbf{h}}= (h_1, \cdots, h_{m\gamma \tau}),\mathsf{td}_a,\overline{\sigma}=(\sigma_{k+1}, \cdots, \sigma_{k+m})$)}:}
\begin{itemize}
\item \textbf{ Input:} 
	 A $(t+k\gamma \tau)$-family of polynomials $\overline{\mathbf{a}}=(a_1,\cdots,a_t,a_{t+1},\cdots,a_{t+k\gamma \tau}) \in (\mathbb{Z}_q^{<n}[x])^t \times (\mathbb{Z}_q^{<n+d-1}[x])^{k\gamma \tau}$, its trapdoor $\mathsf{td}_a$ and a $m\gamma\tau$-family of polynomials $\overline{\textbf{h}}=(h_0, \cdots, h_{m\gamma \tau}) \in  (\mathbb{Z}_q^{<n+d-1}[x])^{m\gamma \tau}$, 
	 where $m \geq 1$, and (implicitly) $\overline{\mathbf{g}}=(g_1,\cdots, g_{\gamma \tau}) \in (\mathbb{Z}_q^{<n+d-1}[x])^{\gamma \tau}$ as in \eqref{k39}.

\item \textbf{Output:} The trapdoor $\mathsf{td}_f$ for $\overline{\textbf{f}}=(a_1,\cdots,a_{t+k\gamma\tau}|h_1, \cdots, h_{m\gamma \tau})$.
\item \textbf{Execution:}
\begin{enumerate}
	\item Split $\overline{\mathbf{h}}=(\overline{\mathbf{h}}^{(1)}, \cdots ,\overline{\mathbf{h}}^{(m)})$ where each $\overline{\mathbf{h}}^{(i)}$ is a $\gamma\tau $-family of polynomials. 

	\item $\mathsf{td}^{(1)} \leftarrow \mathsf{td}_a$, $\overline{\mathbf{a}}^{(1)}\leftarrow \overline{\mathbf{a}}$.
	\item For $i=1$ up to $m$ do:
	\begin{itemize}
		\item  $\mathsf{td}^{(i+1)} \leftarrow \mathsf{TrapDel}(\overline{\mathbf{a}}^{(i)},\overline{\mathbf{h}}^{(i)},\mathsf{td}^{(i)},\sigma_{i})$.
		\item $\overline{\mathbf{a}}^{(i+1)}\leftarrow (\overline{\mathbf{a}}^{(i)}, \overline{\mathbf{h}}^{(i)})$.
	\end{itemize}
	\item Return $\mathsf{td}_f=\mathsf{td}^{(m+1)}$.
\end{enumerate}
\end{itemize}

Let us make few observations 
for \textsf{SampleTrap}}.
\subsubsection{Trapdoor  $\mathsf{td}_f$.} From Section \ref{basictrap}, we can easily generalize to see that the output   $\textsf{td}_f$ is $(\overline{\mathbf{r}}^{(1)}, \cdots, \overline{\mathbf{r}}^{(d \gamma \tau)})$ in which for $i \in [d\gamma \tau]$, $\overline{\mathbf{r}}^{(i)}=(r^{(i)}_1, \cdots,r^{(i)}_{t+(k+m-1) \gamma \tau} )$ and $r^{(i)}_j \in  \mathbb{Z}_q^{<n+d-1}[x] $ for $j \in [t]$, and $r^{(i)}_{t+j} \in  \mathbb{Z}_q^{<d}[x] $ for $j \in [(k+m-1)\gamma\tau]$. We can imply that the matrix representation, named $\mathbf{R}^{(k+m-1)}$, for  the trapdoor $\mathsf{td}_f$ has the form  \eqref{k13}, with the last row's index ${t+(k+m-1) \gamma\tau}$.

\subsubsection{Setting Gaussian parameters $\overline{\sigma}=(\sigma_{1}, \cdots, \sigma_{m})$.} 
Note that the algorithm 
$\mathsf{SamplePre}(\overline{\textbf{a}},\mathsf{td}_a, u,\overline{\sigma})$ 
has to satisfy
Condition \eqref{k12} for each $\sigma_i$.
The same condition must hold for
$\mathsf{GenSamplePre}$.
From Equation \eqref{k13},  we can see that 
the trapdoor $\mathsf{td}^{(i+1)}$ in $\mathsf{SampleTrap}$
can be interpreted as a matrix $\mathbf{R}^{(i+1)}$ of dimension $((2d-1)t+(k+i-1)d\gamma \tau)\times (d\gamma \tau)$.
Thus, $\sigma_{i}$ in Equations \eqref{k17} and \eqref{k18}
should satisfy  $\sigma_{i} \geq \omega(\sqrt{\log (d\gamma)})\cdot \sqrt{7(s_1(\mathbf{R}^{(i-1)})^2+1)}, $     
and \begin{equation*}
s_1(\mathbf{R}^{(i-1)}) \leq \sqrt{((2d-1)t+(k+i-1)d\gamma \tau)\cdot (d\gamma \tau)}\cdot \omega(\log n)\cdot \sigma_{i-1}, \text{ where } i\in [m].
\end{equation*}
    
\section{\text{DMPLWE}-based \text{HIBE} in Standard Model} \label{sec4}
In this section, we describe a \text{HIBE} system based on the \text{DMPLWE} problem. 
Our  \text{HIBE} scheme is \text{IND-sID-CPA} secure in the standard model 
and is inspired by the construction of \text{IBE} from \cite{AB09}. 
Note that the authors of \cite{LVV19} use a similar approach.
However, the private key $\textsf{SK}_\textsf{id}$ (with respect 
to an identity $\textsf{id}$) in the standard model \text{IBE} of \cite{LVV19} 
is actually not a trapdoor.
Therefore, it seems difficult to construct \text{HIBE} using this approach. 
In our \text{HIBE} construction, the private key for an identity   
$\textsf{id}=(id_1, \cdots, id_\ell)$ of depth $\ell$ is a 
trapdoor for a family of polynomials, which corresponds to the public key. 
So we can derive the private key for the appended identity  $\textsf{id}|id_{k}=(id_1, \cdots,id_\ell,\cdots id_k) $ using the trapdoor delegation presented in Section \ref{trap},
where $k>\ell$.

\subsection{Construction}  
Our construction, named \textsf{HIBE}, consists of a tuple of algorithms \{\textsf{Setup}, \textsf{Extract}, \textsf{Derive}, \textsf{Encrypt}, \textsf{Decrypt}\}. They are described below.
\begin{description}
\item \underline{\textbf{\textsf{Setup}($1^\lambda,1^n$)}:} On input the security parameter $n$, the maximum depth $\lambda$, perform the following:
\begin{itemize}
\item Set common parameters as follows:  \begin{itemize}
\item $q=q(n)$ be a prime; $ d, k $ be positive integers such that $2d+k \leq n$ and $\frac{n+2d-2}{d}$ is also a positive integer, say $\gamma$, i.e., $d \gamma =n+2d-2$; $\beta:=\lceil \frac{\log_2 n}{2} \rceil$ $, \tau:=\lceil \log_2 q \rceil$, $t$ is a positive integer and let $t'=t+\gamma \tau$, and plaintext space $\mathcal{M}:=\{0,1\}^{<k+2}[x]$..   Note that we will set $t'=m\gamma \tau$ (with $m \geq 2$), that is $t$ is a multiple of $\gamma \tau$ so as to we can apply the trapdoor delegation.
\item For Gaussian parameters used in \textsf{Encrypt}: choose $\overline{\alpha}=(\alpha_1,\cdots , \alpha_\lambda) \in \mathbb{R}^{\lambda}_{>0}$; for Gaussian parameters used in \textsf{Extract} and  \textsf{Derive}: choose $\overline{\Sigma}=(\overline{\sigma}^{(1)},\cdots , \overline{\sigma}^{(\lambda)})$, where $\overline{\sigma}^{(\ell)}=(\sigma^{(\ell)}_1, \cdots,$ $ \sigma^{(\ell)}_m) \in \mathbb{R}^{m}_{>0}$. For $\ell \in [\lambda]$, let $\overline{\Sigma}^{(\ell)}=(\overline{\sigma}^{(1)},\cdots , \overline{\sigma}^{(\ell)})$; for Gaussian parameters used in \textsf{Decrypt}: choose $\overline{\Psi}=(\Psi_1, \cdots, \Psi_{\lambda})\in \mathbb{R}^{\lambda}_{>0}$.
\end{itemize} 
They all are set as in Section \ref{params}.
\item For $\ell \in [\lambda]$, let  $\chi_\ell:=\lfloor D_{\alpha_\ell \cdot q}\rceil$ be  the rounded Gaussian distribution.
\item Use $\mathsf{TrapGen}(1^{n})$  to get a root family  $\overline{\mathbf{a}}_{\epsilon}=(a_1,\cdots,a_{t'})$ and its associated root trapdoor $\mathsf{td}_\epsilon$. 
\item Select uniformly a random polynomial $u_0 \in \mathbb{Z}_q^{<n+2d-2}[x]$. 

\item For each $i\in [\lambda]$, and each $\mathsf{bit}\in\{0,1\}$, sample randomly $\overline{\mathbf{h}}^{(i,\mathsf{bit})}=(h_1^{(i,\mathsf{bit})}, \cdots,h_{t'}^{(i,\mathsf{bit})} )$, where each $h_{j}^{(i,\mathsf{bit})}\in \mathbb{Z}_q^{<n}[x]$ for $j\in [t]$, and each $h_{j}^{(i,\mathsf{bit})}\in \mathbb{Z}_q^{<n+d-1}[x]$  for $j\in \{t+1, \cdots, t+\gamma\tau\}$. Let $\textsf{HList}=\{(i,\textsf{bit}, \overline{\mathbf{h}}^{(i,\mathsf{bit})}): i \in [\lambda], \textsf{bit}\in \{0,1\} \}$ be the ordered set of all $\overline{\mathbf{h}}^{(i,\mathsf{bit})}$. 

\item Set the master secret key $\mathsf{MSK}:=\mathsf{td}_\epsilon$. 
\end{itemize}

We denote $\textsf{id}=(id_1, \cdots, id_\ell) \in \{0,1\}^\ell$ as an identity of depth $\ell \leq \lambda$. All following algorithms will always work on $\overline{\mathbf{a}}_{\epsilon}=(a_1,\cdots,a_{t'})$ and $\textsf{HList}$.

\item \underline{\textbf{\textsf{Derive}}$( \mathsf{id}|id_{\ell+1},\mathsf{SK}_{\textsf{id}}):$} On input  $\mathsf{id}=(id_1, \cdots, id_\ell)$, $\mathsf{id}|id_{\ell+1}=(id_1,$ $ \cdots, id_\ell, id_{\ell+1})$,  private key  $\mathsf{SK}_\mathsf{id}:=\textsf{td}_{\textsf{id}}$-- the trapdoor for $\overline{\mathbf{f}}_{\mathsf{id}}=(\overline{\mathbf{a}}_{\epsilon}, \overline{\mathbf{h}}^{(1,id_1)}$, $  \cdots,\overline{\mathbf{h}}^{(\ell,id_{\ell})})$, execute: 
\begin{enumerate}
\item Output $\mathsf{SK}_{\mathsf{id}|id_{\ell+1}} \leftarrow \mathsf{SampleTrap}(\overline{\mathbf{f}}_{\mathsf{id}},\overline{\mathbf{h}}^{(\ell+1,id_{\ell+1})},\mathsf{SK}_\mathsf{id},\overline{\Sigma}^{(\ell +1)})$.
\end{enumerate}

\item \underline{ \textbf{$\textsf{Extract}(\mathsf{id},\mathsf{MSK})$}:} On input   $\mathsf{id}=(id_1, \cdots, id_\ell)$,  $\mathsf{MSK}=\mathsf{td}_\epsilon$, execute:

\begin{enumerate}

\item Build $\overline{\mathbf{h}}_\mathsf{id}=( \overline{\mathbf{h}}^{(1,id_1)},  \cdots,\overline{\mathbf{h}}^{(\ell,id_{\ell})})$.
\item Output $\mathsf{SK}_{\textsf{id}} \leftarrow \mathsf{SampleTrap}(\overline{\mathbf{a}}_{\epsilon},\overline{\mathbf{h}}_{\mathsf{id}},\mathsf{MSK},\overline{\Sigma}^{(\ell)})$.
\end{enumerate}

\item \underline{\textbf{\textsf{Encrypt}}($\mathsf{id},\mu, u_0, \alpha_\ell$):} On input $\mathsf{id}=(id_1, \cdots, id_\ell)$, $ \mu \in \mathcal{M}$, $u_0$, $\alpha_\ell$, execute:
\begin{enumerate}
\item Build $(f_1, \cdots, f_{t'(\ell+1)}) \leftarrow \overline{\mathbf{f}}_{\mathsf{id}}=(\overline{\mathbf{a}}_{\epsilon}, \overline{\mathbf{h}}^{(1,id_1)} \cdots,\overline{\mathbf{h}}^{(\ell,id_{\ell})}).$
\item Sample $s \xleftarrow[]{\$}\mathbb{Z}_q^{<n+2d+k-1}[x]$. 
\item Sample  $e_0 \leftarrow \chi_{\ell}^{k+1}[x]$, compute: $\mathsf{CT}_{0}=u_0\odot_{k+2}s+2e_0+\mu$.
\item For $i=0$ to $ \ell$ do: 
\begin{itemize}
\item For $j \in [t]$, sample $e_{i\cdot t'+j} \leftarrow \chi_{\ell}^{2d+k}[x]$, and compute:
 $$\textsf{ct}_i=f_{i\cdot t'+j}\odot_{2d+k}s+2e_{i\cdot t'+j}.$$
\item For $t+1 \leq j \leq t+\gamma \tau$, sample $e_{i\cdot t'+j} \leftarrow \chi_{\ell}^{d+k+1}[x]$, and compute: $$\textsf{ct}_i=f_{i\cdot t'+j}\odot_{d+k+1}s+2e_{i\cdot t'+j}.$$
\end{itemize}

\item Set $\mathsf{CT}_{1}=(\textsf{ct}_1, \cdots, \textsf{ct}_{t'(\ell+1)})$, and output ciphertext $\overline{\mathsf{CT}}=(\mathsf{CT}_0, \mathsf{CT}_1)$.
 
\end{enumerate}

\item \underline{\textbf{\textsf{Decrypt}}($\mathsf{id}, \mathsf{SK}_\mathsf{id},\overline{\mathsf{CT}},u_0, \Psi_\ell $):}
On input  $\mathsf{id}=(id_1, \cdots, id_\ell)$,   $\mathsf{SK}_\mathsf{id}:=\textsf{td}_\textsf{id}$--the trapdoor  for $ \overline{\mathbf{f}}_{\mathsf{id}}=(\overline{\mathbf{a}}_{\epsilon},   \overline{\mathbf{h}}^{(1,id_1)}$, $  \cdots,\overline{\mathbf{h}}^{(\ell,id_{\ell})})$, ciphertext $\overline{\mathsf{CT}}=(\mathsf{CT}_0, \mathsf{CT}_1)$, $u_0$, and $\Psi_\ell $,  do:

\begin{enumerate}
\item Parse  $(f_1, \cdots, f_{t'(\ell+1)}) $ $\leftarrow \overline{\mathbf{f}}_{\mathsf{id}}=(\overline{\mathbf{a}}_{\epsilon},   \overline{\mathbf{h}}^{(1,id_1)}$, $  \cdots,\overline{\mathbf{h}}^{(\ell,id_{\ell})})$.
\item Sample $\overline{\textbf{r}}=(r_1,\cdots, r_{t'(\ell+1)}) \leftarrow \mathsf{GenSamplePre}(\overline{\textbf{f}}_{\mathsf{id}}, \mathsf{SK}_\mathsf{id}, u_0,\Psi_\ell)$, \\i.e., $\langle \overline{\textbf{f}}_{\mathsf{id}}, \overline{\textbf{r}} \rangle =\sum_{1}^{t'(\ell+1)}r_i\cdot f_i=u_0$.
\item Parse $(\mathsf{CT}_0, \mathsf{CT}_1=(\textsf{ct}_1, \cdots, \textsf{ct}_{t'(\ell+1)})) \leftarrow \overline{\mathsf{CT}}$.
\item Output $\mu=(\mathsf{CT}_{0}-\sum_{i=1}^{t'(\ell+1)} \textsf{ct}_{i}\odot_{k+2} r_i \mod q) \mod 2$.
\end{enumerate}

\end{description}

\subsection{Correctness and Parameters} \label{params}
\begin{lemma}[Correctness] For $\ell \in [\lambda]$, if 
\begin{equation}\label{k21}
	\alpha_\ell <\frac{1}{4} \left[t'(\ell+1)\cdot(k+1)\cdot \omega(\log n)\cdot \Psi_\ell   +\omega(\sqrt{\log n})\right]^{-1},
\end{equation}
	 then the scheme is correct with probability $1-\mathsf{negl}(n)$.
\end{lemma}
\begin{proof}
For $\textsf{id}=(id_1, \cdots, id_\ell)$, we need to show that 
$$\textsf{Decrypt}(\mathsf{id}, \mathsf{SK}_\mathsf{id},\textsf{Encrypt}(\mathsf{id},\mu, u_0, \alpha_\ell),u_0, \Psi_\ell)=\mu,$$
 with probability $1-\textsf{negl}(n)$ over the randomness of \textsf{Setup}, \textsf{Derive}, \textsf{Extract}, \textsf{Encrypt}.
Suppose that  $\overline{\mathsf{CT}}:=(\mathsf{CT}_0, \mathsf{CT}_1=(\textsf{ct}_1, \cdots, \textsf{ct}_{t'(\ell+1)})) \leftarrow \textsf{Encrypt}(\mathsf{id},\mu, u_0, \alpha_\ell).$ By Lemma \ref{lem3}, we have

 $\mathsf{CT}_{0}-\sum_{i=1}^{t'(\ell+1)} \textsf{ct}_{i}\odot_{k+2} r_i =\mu+2(e_0-\sum_{1}^{t'(\ell+1)}r_i\odot_{k+2} e_i)$. \\
  Hence, 
if $\Vert \mu+2(e_0-\sum_{1}^{t'(\ell+1)}r_i\odot_{k+2} e_i)\Vert_{\infty}< q/2$ then $\mu $ is recovered.

Therefore, we need to bound the coefficients of $e_0-\sum_{1}^{t'(\ell+1)}r_i\odot_{k+2} e_i$. First, note that, 
\begin{itemize}
\item for $i \in [t]$: $\deg(r_i)< d_r:=2d-1$, $\deg(e_i)<d_e:=k+1$. 
\item for $i \in \{ t+1, \cdots, t'(\ell+1)\}$: $\deg(r_i)< d_r:=d$, $\deg(e_i)<d_e:=d+k+1$.
\end{itemize}
In gerneral, $d_e+d_r-1=2(d-1)+(k+2)$. Let $\textbf{r}_i=(\textbf{r}_{i,0}, \cdots, \textbf{r}_{i,d_r-1}), \textbf{e}_i=(\textbf{e}_{i,0}, \cdots, \textbf{e}_{i,d_e-1})$ be the vectors of coefficients of $r_i$ and $e_i$, respectively. By definition of the middle product,
$r_i\odot_{k+2} e_i= \sum_{j+w=d-1}^{d+k}\textbf{r}_{i,j}\cdot \textbf{e}_{i,w}\cdot x^{j+w}.$
By Lemma \ref{lem4},
$\Pr[\Vert \textbf{r}_i\Vert_{\infty} > \omega(\sqrt{\log n})\cdot \Psi_\ell ]=\textsf{negl}(n), $
$\Pr[\Vert \textbf{e}_i\Vert_{\infty}> \omega(\sqrt{\log n})\cdot \alpha_{\ell}\cdot q]=\textsf{negl}(n). $

Hence $\Vert r_i\odot_{k+2} e_i\Vert_{\infty} <(k+2)\cdot \omega(\log n)\cdot \Psi_\ell \cdot \alpha_{\ell}\cdot q.$
 As a result, 
$$\left\Vert e_0-\sum_{1}^{t'(\ell+1)}r_i\odot_{k+2} e_i\right \Vert_{\infty}\leq [t'(\ell+1)\cdot(k+2)\cdot \omega(\log n)\cdot \Psi_\ell  +\omega(\sqrt{\log n})]\cdot \alpha_{\ell}\cdot q.$$
In order for the decryption to be correct, we need Condition \eqref{k21}.
\qed
\end{proof}

\subsubsection{Setting Parameters.}  We set the parameters as described below:
 \begin{itemize} 
 	\item Security prameter $n$, $q=\textsf{poly}(n)$ prime, $\beta:=\lceil \frac{\log_2(n)}{2} \rceil \ll q/2$, $ \tau:=\lceil \log_2 (q) \rceil$, $\tau=\Theta(\log q)=\Theta(\log n)$,  $t'=t+\gamma \tau=m\gamma \tau$ (for some $m \geq 2$), $d \leq n$, $dt/n=\Omega(\log n)$, and $d\gamma=n+2d-2 \leq 3n$.
 	\item  
 	We set Gaussian parameters  used in \textsf{Extract} and \textsf{Derive} as follows: Recall that, for $\ell \in [\lambda]$, $\overline{\Sigma}^{(\ell)}=(\overline{\sigma}^{(1)},\cdots , \overline{\sigma}^{(\ell)})$, where each $ \overline{\sigma}^{(i)}=(\sigma^{(i)}_1, \cdots, \sigma^{(i)}_m) \in \mathbb{R}^{m}_{>0}$.  It suffices to consider the maximal case happening in \textsf{Extract} in which $\overline{\Sigma}=(\overline{\sigma}^{(1)},\cdots , \overline{\sigma}^{(\lambda)})$. Now, we renumber $\overline{\Sigma}$ as $(\sigma_1, \cdots, \sigma_{m\lambda})$ without changing their order. 
For the maximal identity $\textsf{id}=(id_1, \cdots, id_\lambda)$, we build $\overline{\mathbf{h}}_\mathsf{id}=( \overline{\mathbf{h}}^{(1,id_{1})}, $ $ \cdots,\overline{\mathbf{h}}^{(\lambda,id_\lambda)})$ and then compute  $\mathsf{SK}_{\textsf{id}}$ by calling $ \mathsf{SampleTrap}$ for input $(\overline{\mathbf{a}}_{\epsilon},\overline{\mathbf{h}}_{\mathsf{id}},$ $\mathsf{MSK},\overline{\Sigma})$. 
We now split $\overline{\mathbf{h}}_\mathsf{id}$ into $(\overline{\mathbf{h}}^{(1)}, \cdots, \overline{\mathbf{h}}^{(m\lambda)})$ and let $\overline{\mathbf{a}}^{(i)}=(\overline{\mathbf{a}}_{\epsilon}|\overline{\mathbf{h}}^{(1)}|\cdots| \overline{\mathbf{h}}^{(i)})$ with $\overline{\mathbf{a}}^{(0)}=\overline{\mathbf{a}}_{\epsilon}$. 
Then, $ \mathsf{SampleTrap}$ calls \textsf{TrapDel}$(\overline{\mathbf{a}}^{(i-1)},\overline{\mathbf{h}}^{(i)},$ $\mathsf{td}^{(i-1)},\sigma_{i})$
 up to $m \lambda$ times 
 for $i\in [m\lambda]$, in which  $\mathsf{td}^{(0)}=\mathsf{td}_{\epsilon}$ and 
 $\mathsf{td}^{(i-1)}$ is the output of the previous execution of \textsf{TrapDel}$(\overline{\mathbf{a}}^{(i-2)},\overline{\mathbf{h}}^{(i-1)},\mathsf{td}^{(i-2)},\sigma_{i-1})$, for $2 \leq i \leq m\lambda$. 
 Now, all $\sigma_i$'s are set in the same way as in Section \ref{multrap}, that is, for $2 \leq i \leq m \lambda$ , 
 	$\sigma_i \geq \omega(\sqrt{\log (d\gamma)})\cdot \sqrt{7(s_1(\mathbf{R}^{(i-1)})^2+1)}, $ and
 	 \begin{equation*}\label{k23}
 	s_1(\mathbf{R}^{(i-1)}) \leq \sqrt{((2d-1)t+(i-1)d\gamma \tau)\cdot (d\gamma \tau)}\cdot \omega(\log n)\cdot \sigma_{i-1},
 	\end{equation*}

 	in which $\mathbf{R}^{(i-1)}$ is the matrix representation, as in \eqref{k13} with the last row's index  $t+(i-1)\gamma\tau$, of the private key (the trapdoor) for $\overline{\mathbf{a}}^{(i-1)}=(\overline{\mathbf{a}}_{\epsilon}|\overline{\mathbf{h}}^{(1)}|\cdots| \overline{\mathbf{h}}^{(i-1)})$, with  $\sigma_1$ and $\mathbf{R}^{(1)}$ play the role of $\sigma$ and $\mathbf{T}$ in \eqref{k12}, \eqref{k19}.
 
 		\item  We set Gaussian parameters used in \textsf{Decrypt} $\overline{\Psi}=(\Psi_1,\cdots , \Psi_\lambda)$ as follows: For $\ell \in [\lambda]$, since  $\Psi_\ell$ is used in  $\mathsf{GenSamplePre}(\overline{\textbf{f}}_{\mathsf{id}}, \mathsf{SK}_\mathsf{id}, u_0,\Psi_\ell)$ with $\mathsf{SK}_\mathsf{id}=\textsf{td}_\textsf{id}$ the trapdoor for $\overline{\mathbf{f}}_{\mathsf{id}}=(\overline{\mathbf{a}}_{\epsilon}, \overline{\mathbf{h}}^{(1,id_1)}$, $  \cdots,\overline{\mathbf{h}}^{(\ell,id_{\ell})})$ which equals to $\overline{\mathbf{a}}^{(\ell m)}$ above. Therefore, for $\ell \in [\lambda-1]$ we can set
 	$\Psi_\ell =\sigma_{\ell m+1},$ and 
 	 	
 	 	\begin{equation*}\label{k36}
 	 	 	\Psi_\lambda \geq \omega(\sqrt{\log (d\gamma)})\cdot \sqrt{7(s_1(\mathbf{R}^{(m \lambda)})^2+1)}, 
 	 	 	\end{equation*}
 	 	 	 \begin{equation*}\label{k37}
 	 	 	s_1(\mathbf{R}^{(m \lambda)}) \leq \sqrt{((2d-1)t+(m \lambda)d\gamma \tau)\cdot (d\gamma \tau)}\cdot \omega(\log n)\cdot \sigma_{m \lambda},
 	 	 	\end{equation*}
 in which $\mathbf{R}^{(m \lambda)}$ is the matrix representation for the private key (the trapdoor) for $\overline{\mathbf{a}}^{(m \lambda)}=(\overline{\mathbf{a}}_{\epsilon}|\overline{\mathbf{h}}^{(1)}|\cdots| \overline{\mathbf{h}}^{(m \lambda)})$.
 		\item We set Gaussian parameters used in \textsf{Encrypt} $\overline{\alpha}=(\alpha_1,\cdots , \alpha_\lambda)$ such that for $\ell \in [\lambda]$, $\alpha_{\ell}$ satisfies \eqref{k21}.
 \end{itemize}

\subsection{Security Analysis}
\begin{theorem}

The proposed $\mathsf{HIBE}$ system is IND-sID-CPA secure in the standard model under the $\mathsf{DMPLWE}$ assumption.
\end{theorem}
\proof 
We construct a sequence of games from $G_0$ to $G_4$ in which an \text{INDr--sID--CPA} adversary can distinguish two consecutive games $G_i$ and $G_{i+1}$ \textit{with negligible probability only}. In particular, for the transition of the last two games $G_3$ and $G_4$, we show by contradiction that if there exists an adversary whose views are different in each game, i.e., the adversary can distinguish $G_3$ from $G_4$ with non-negligible probability,  then we can build an adversary who can solve the underlying $\mathsf{DMPLWE}$ problem.\\ 

\noindent\underline{\textbf{Game $G_0$}} is the original $\text{IND--sID--CPA}$ game between the adversary $\mathcal{A}$ and the challenger $\mathcal{C}$. 
Note that, we are working with the selective game: at the beginning, $\mathcal{A}$ lets the challenger know the target identity $\mathsf{id}^*=(id^*_1, \cdots, id^*_\theta)$ that $\mathcal{A}$ intends to atack, where $\theta \leq \lambda$. Then, $\mathcal{C}$ runs \textsf{Setup} to choose randomly a vector of polynomials $\overline{\textbf{a}}_{\epsilon}=(a_1, \cdots, a_{t'})$ together with an associated trapdoor $\mathsf{td}_{\epsilon}$, a set of polynomial vectors sampled randomly $\overline{\mathbf{h}}^{(i,\mathsf{bit})}=(h_1^{(i,\mathsf{bit})}, \cdots,h_{t'}^{(i,\mathsf{bit})} )$,   which are stored in \textsf{HList0}, where each $h_{j}^{(i,\mathsf{bit})}\in \mathbb{Z}_q^{<n}[x]$ for $j\in [t]$, and each $h_{j}^{(i,\mathsf{bit})}\in \mathbb{Z}_q^{<n+d-1}[x]$  for $j\in \{t+1, \cdots, t+\gamma\tau\}$, and $\mathcal{C}$ also chooses a random polynomial $u_0 \in \mathbb{Z}_q^{<n+2d-2}[x]$. The challenger then sets $\mathsf{MSK}:=\mathsf{td}_{\epsilon}$ as the master secret key. Furthermore,  at the Challenge Phase, the challenger also generates a challenge ciphertext $\overline{\mathsf{CT}}^*$ for the identity $\mathsf{id}^*$.\\

\noindent\underline{\textbf{Game $G_1$}} is the same as $G_0$ except that in the Setup Phase the challenger $\mathcal{C}$ generates  $ (\overline{\textbf{h}}^{(i,\mathsf{bit})})_{0\leq i \leq \lambda, \mathsf{bit}\in \{0,1\}}$ stored in $\textsf{HList1}:=\{(i,\textsf{bit}, \overline{\textbf{h}}^{(i,\mathsf{bit})}): i \in [\lambda], \textsf{bit}\in \{0,1\} \}$ with the corresponding trapdoor $\textsf{td}^{(i, \textsf{bit})}$ stored in  $\textsf{TList1}:=\{(i,\textsf{bit}, \mathsf{td}^{(i,\mathsf{bit})}): i \in [\lambda], \textsf{bit}\in \{0,1\} \}$ using \textsf{TrapGen}.\\

\noindent\underline{\textbf{Game $G_2$}} is the same as $G_1$, except that the challenger $\mathcal{C}$ does not use $\textsf{td}_{\epsilon}$ as the master secret key nor the \textsf{Extract} algorithm to response a private key queries on $\textsf{id}=(id_1, \cdots, id_{\ell})$ which is not a prefix of the target $\textsf{id}^*$, where $\ell \leq \lambda$. Instead, $\mathcal{C}$ designs a new procedure \textsf{TrapExtract} with the knowledge of \textsf{TList1}. \textsf{TrapExtract} requires not all of $\textsf{TList1}$ but only one  $\mathsf{td}^{(j,id_{j})} \in \textsf{TList1}$ for any $j \in [\ell]$.

\underline{\textsf{TrapExtract}($\overline{\textbf{a}}_{\epsilon}$,\textsf{HList1},$\textsf{id}=(id_1, \cdots, id_{\ell})$, $j$,$\textsf{td}^{(j,id_{j})}$):}  

\begin{enumerate}

\item Build $\overline{\mathbf{f}}_{\mathsf{id}}=(\overline{\mathbf{a}}_{\epsilon}, \overline{\mathbf{h}}^{(1,id_1)},  \cdots,\overline{\mathbf{h}}^{(j-1,id_{j-1})}, \overline{\mathbf{h}}^{(j+1,id_{j+1})}, \cdots,\overline{\mathbf{h}}^{(\ell,id_{\ell})})$
\item $\mathsf{SK}_{\textsf{id}} \leftarrow \mathsf{SampleTrap}(\overline{\mathbf{h}}^{(j,id_j)},\overline{\mathbf{f}}_{\mathsf{id}},\textsf{td}^{(j,id_{j})},\overline{\Sigma}^{(\ell)})$
\end{enumerate}

\noindent\underline{\textbf{Game $G_3$}} is the same as $G_2$, except that in the Setup Phase, the challenger $\mathcal{C}$ generates $\textsf{HList3}$ as follows: 
\begin{itemize}
\item For each $j \in [\lambda]$ and $\textsf{bit}\in\{0,1\}$ such that $\textsf{bit} \neq  id^*_j,$ $\mathcal{C}$ calls \textsf{TrapGen} to generate  $ \overline{\textbf{h}}^{(j,\mathsf{bit})}$ and its associated trapdoor $\textsf{td}^{(j,\textsf{bit})}$.  

\item For each $j \in [\lambda]$ and $\textsf{bit}\in\{0,1\}$ such that $\textsf{bit} =  id^*_j,$ $\mathcal{C}$ simply samples $ \overline{\textbf{h}}^{(j,\mathsf{bit})}$ uniformly at random and set $\textsf{td}^{(j,\textsf{bit})}=\bot$.  
 
\end{itemize}
The challenger then put all $ \overline{\textbf{h}}^{(j,\mathsf{bit})}$ into \textsf{HList3} and all $\textsf{td}^{(j,\textsf{bit})}$ into \textsf{TList3}.
 At the moment, to response a private key query on identity $\textsf{id}=(id_1, \cdots, id_{\ell})$ which is not a prefix of the target identity $\textsf{id}^*$, the challenger chooses an index $j^{\dagger}$ such that $id_{j^{\dagger}}\neq id^*_{j^{\dagger}}$. It then runs \textsf{TrapExtract}($\overline{\textbf{a}}_{\epsilon}$,\textsf{HList3},$\textsf{id}=(id_1, \cdots, id_{\ell})$, $j^{\dagger}$,$\textsf{td}^{(j^{\dagger},id_{j^{\dagger}})}$), where $\textsf{td}^{(j^\dagger,id_{j^\dagger})}  \in \textsf{TList3}$, and gives the result $\textsf{SK}_{\textsf{id}}$ to the adversary.
At the Challenge Phase, the challenge ciphertext $\overline{\textsf{CT}}^*$ is generated by computing \textsf{Encrypt}($\mathsf{id},\mu, u_0, \overline{\alpha}$) over $\mathsf{HList3}$.\\

\noindent\underline{\textbf{Game $G_4$}}  is the same as $G_3$, 
except that the challenge ciphertext $\overline{\textsf{CT}}^*=(\textsf{CT}^*_0, \textsf{CT}^*_1)$ 
is chosen uniformly at random by the challenger. 

In what follows, we show the indistinguishability of the games.
It is easy to see that the view of the adversary is identical in games $G_0$ and $G_1$, 
in games $G_1$ and $G_2$, in games $G_2$ and $G_3$, except in games $G_3$ and $G_4$. 
We show that the view of the adversary is indistinguishable in these two games. 
We proceed by contradiction.
Assume that the adversary $\mathcal{A}$ can distinguish between games $G_3$ and $G_4$ 
with non-negligile probability. 
Then we construct an adversary $\mathcal{B}$ that is able to solve \textsf{DMPLWE} problem with the same probability. The reduction from \textsf{DMPLWE} is as follows:
\begin{itemize}
\item \textbf{Instance:} Assume that the goal of $\mathcal{B}$ is to decide whether $1+t'(\ell+1)$ samples $ (f_{z},\textsf{ct}_{z})$ for $z \in \{0, \cdots, t'(\ell+1)\}$  (i) follow $ \prod_{z=0}^{t'(\ell+1)} \mathcal{U}(\mathbb{Z}_q^{n'-d_{z}}[x] \times \mathbb{R}_q^{d_{z}[x]} ),$ or (ii) follow $ \mathsf{DMP}_{q,n',\textbf{d}, \chi}(s)$, where $n'=n+2d+k$ and 
\begin{itemize}
\item 
$\textbf{d}=(d_0, d_1,\cdots, d_{t'(\ell+1)})$ is interpreted as follows: $d_0:=k+2$ and for $i\in \{0, \cdots, \ell\}$,
 $d_{i\cdot t'+j} = \begin{cases} 
 2d+k, & \mbox{if } j \in [t] ,\\ 
 d+k+1, & \mbox{if } j \in \{t+1, \cdots, t+\gamma \tau \}. 
  \end{cases}$

\item  $f_z$ are random in $\mathbb{Z}_q^{<n'-d_z}[x]$ for $z\in \{0, \cdots, t'(\ell+1)\}$.

\end{itemize}
In other words, $\mathcal{B}$ has to distinguish whether (i) all $\textsf{ct}_{z}$ are random or (ii) $\textsf{ct}_z=f_z\odot_{d_z}s+2e_z$ in $\mathbb{Z}_q^{<d_z}[x]$, for some $s \xleftarrow[]{\$}\mathbb{Z}_q^{<n'-1}[x]$ and 
 $e_{z} \leftarrow \chi^{d_{z}}[x]$, for all $z\in \{0, \cdots, t'(\ell+1)\}$.

\item \textbf{Targeting:} $\mathcal{B}$ receives from the adversary $\mathcal{A}$ the target identity $\textsf{id}^*$ that $\mathcal{A}$ wants to attack.
\item \textbf{Setup:}. $\mathcal{B}$ generates \textsf{HListB} in the same way as in Game $G_3$ and Game $G_4$ as follows:
\begin{itemize}
\item For each $j \in [\lambda]$ and $\textsf{bit}\in\{0,1\}$ such that $\textsf{bit} \neq  id^*_j:$ $\mathcal{B}$ calls \textsf{TrapGen} to generate  $ \overline{\textbf{h}}^{(j,\mathsf{bit})}$ and its associated trapdoor $\textsf{td}^{(j,\textsf{bit})}$.  

\item For each $j \in [\lambda]$ and $\textsf{bit}\in\{0,1\}$ such that $\textsf{bit} =  id^*_j:$ $\mathcal{B}$ simply samples $ \overline{\textbf{h}}^{(j,\mathsf{bit})}$ uniformly at random and set $\textsf{td}^{(j,\textsf{bit})}=\bot$.  
 
\end{itemize}
The challenger then put all $ \overline{\textbf{h}}^{(j,\mathsf{bit})}$ into \textsf{HListB} and all $\textsf{td}^{(j,\textsf{bit})}$ into \textsf{TListB}.
\item \textbf{Queries}: To response the private key queries, $\mathcal{B}$ acts as in Game $G_3$ or in Game $G_4$ using one of trapdoors that is not $\bot$.
\item \textbf{Challenge:} To produce the challenge ciphertext, $\mathcal{B}$ chooses randomly $b\xleftarrow[]{\$}\{0,1\}$ and sets $ \overline{\mathsf{CT}}^*:=(\mathsf{CT}^*_0:=ct_0+\mu, \mathsf{CT}^*_1:=(\textsf{ct}_1, \cdots, \textsf{ct}_{t'(\ell+1)})) $.
\item  \textbf{Guess:} Eventually, $\mathcal{A}$ has to guess and output the value of $b$. Then, $\mathcal{B}$ returns what $\mathcal{A}$ outputted. 
\end{itemize}

\textbf{Analysis. }Clearly, from the view of $\mathcal{A}$, the behaviour of $\mathcal{B}$ is almost identical in both Games $G_3$ and $G_4$. The only different thing is producing the challenge ciphertext. Specifically, if $\textsf{ct}_z$'s are \textsf{DMPLWE} samples then the components of $\overline{\mathsf{CT}}^*$ are distributed as in Game $G_3$, while $\textsf{ct}_z$'s are random then the components of $\overline{\mathsf{CT}}^*$ are distributed as  in Game $G_4$. Since  $\mathcal{A}$ can distinguish between Games $G_3$ and $G_4$ with non-negligile probability, then so can $\mathcal{B}$ in solving \textsf{DMPLWE} with the same probability. \qed

\section{Conclusions} \label{sec5}
In this paper, we present a trapdoor delegation method that enables us to obtain a trapdoor for an expanded set of polynomials from a given trapdoor for a subset of the set. Also, thanks to the polynomial trapdoor delegation, we built a hierarchical identity--based encryption system that is secure in the standard model under the \textsf{DMPLWE} assumption.  
\subsubsection{Acknowledgment.} 

We all would like to thank anonymous reviewers for their helpful comments.  This work is partially supported by the Australian Research Council Discovery Project DP200100144. The first author has been sponsored by a Data61 PhD Scholarship.  The fourth author has been supported by the Australian ARC grant DP180102199 and Polish NCN grant 2018/31/B/ST6/03003.

\bibliographystyle{splncs04}

\end{document}